\documentclass[12pt]{article}
\usepackage{titlesec}
\usepackage{ntheorem}
\titlespacing*{\section}{0pt}{1pt}{1pt}
\titlespacing*{\subsection}{0pt}{1pt}{1pt}
\titlespacing*{\subsubsection}{0pt}{0pt}{0pt}
\usepackage{pdfpages}
\usepackage{caption}
\usepackage{subcaption}
\usepackage{colortbl}
\usepackage{float}
\usepackage{booktabs}
\usepackage{bbm}
\usepackage[normalem]{ulem}
\usepackage{charter}
\usepackage{pgf}
\usepackage{tikz}
\usetikzlibrary{backgrounds, calc, arrows.meta}
\usetikzlibrary{arrows,automata}
\usepackage[latin1]{inputenc}
\usepackage{verbatim}
\usepackage{xcolor}
\usepackage{natbib}
\colorlet{drkblue}{blue!61.8!black}
\usepackage[bookmarks=false]{hyperref}
\hypersetup{
     pdfborder = {0 0 0},
  colorlinks=true,
 citecolor=drkblue,
linkcolor=drkblue,
urlcolor=drkblue,
pdfstartview={XYZ null null 0.80}
}
\usepackage[hyperpageref]{backref}
\usepackage{amsmath}
\usepackage{amssymb}
\usepackage[left=2.5cm,top=3cm,bottom=3.5cm,right=3cm]{geometry}
\usetikzlibrary{math}

\newtheorem{theorem}{Theorem}
\newtheorem{proposition}{Proposition}
\newtheorem{lemma}{Lemma}[section]

\newtheorem{corollary}{Corollary}

\newtheorem{example}{Example}

\theorembodyfont{\rmfamily}
\theorembodyfont{\rmfamily}
\newenvironment{proof}[1][Proof:]{\noindent\textbf{#1} }

\title{Belief Identification in Populations}

 \author{\hspace*{0cm} Christopher P. Chambers\thanks{Department of Economics, Georgetown University, ICC 580  37th and O Streets NW, Washington DC 20057. E-mail: \texttt{Christopher.Chambers@georgetown.edu}.}   \and Yusufcan Masatlioglu\thanks{University of Maryland, 3147E Tydings Hall, 7343 Preinkert Dr.,  College Park, MD 20742. E-mail: \texttt{yusufcan@umd.edu.}}
   \and R. Emilio Mu\~{n}iz-Langle\thanks{Department of Economics, Georgetown University, ICC 580  37th and O Streets NW, Washington DC 20057. E-mail:\texttt{rm2104@georgetown.edu.}}
  }
 
\date{\today}

\begin{document}

\maketitle

\begin{abstract}

We study the identification of belief distributions in a population of Bayesian agents from anonymous aggregate belief data. While a single Bayesian agent's full belief can be recovered from beliefs over a suitable collection of binary events, this principle need not extend to populations: event-by-event distributions of beliefs may fail to identify the underlying distribution of priors. We study when this failure is generic and when it is exceptional. Identification is governed by the graph-theoretic structure induced by the observed family of events on the state space. Among $n$-agent distributions, identification is generic if the induced graph is nonseparable, while non-identification is generic if the graph is separable. The results establish both limits and design principles for recovering belief heterogeneity from aggregate belief data.

   (JEL: D01, D91)

\end{abstract}

%%==========================================================================
%% SETUP AND NOTATION
%%==========================================================================
\section{Introduction}

Beliefs play a crucial role in economics because they shape how individuals, firms, and markets make decisions under uncertainty (\cite{manski2004}). Our beliefs about the future (e.g., prices, interest rates, inflation, asset values, risk, and returns) influence spending, saving, and investment decisions (\cite{gennaioli2018}).  In this sense, beliefs are a central, though often unobserved, determinant of economic behavior (\cite{savage1971}).

Understanding what people believe is essential for policy making. Many policies are designed under uncertainty, but their success often depends on how people themselves perceive that uncertainty. A vaccination campaign, for example, may fail not because vaccines are unavailable, but because people have different beliefs about the risks of disease or the safety of treatment. Similar issues arise in education, insurance, job search, financial decisions, and responses to climate risk. In all of these cases, policy makers need to know not only what incentives people face, but also what beliefs guide their decisions (\cite{manski2018}).

An individual belief over a state space can often be inferred from conditional beliefs over subclasses of events.% Under Bayesian consistency, relative beliefs over sufficiently many binary subsets of states can be combined to recover the full belief distribution. 
Specifically, when a Bayesian reports the relative likelihood of one state to another, the relative likelihoods can, under some conditions, be used to uniquely back out the underlying belief distribution. The main condition, under an assumption of full support beliefs, required for this inverse procedure is that the pairs of states present in some event in the subclass form a \emph{connected graph}.  This is important because reporting a full probability distribution can be difficult, especially when the state space is large or complicated. Binary elicitation is typically easier to implement, less cognitively demanding, and more natural in many applications (\cite{CHARNESS2021}).

This paper asks whether an analogous principle holds at the population level. Suppose agents have heterogeneous beliefs, and suppose the analyst does not observe any individual's full beliefs. Instead, for each event in some family, the analyst observes the distribution of Bayesian updates induced by the event across the population. Importantly, the data are anonymous: the analyst observes the distribution of updated beliefs for each event, but cannot track which beliefs belong to the same individual across different events. This anonymity creates an identification problem. At the individual level, the analyst can link a person's answers across questions; at the population level, those links are missing.  Our primary question is to understand under which conditions the analyst can back out the underlying distribution of beliefs.

Answering this question has direct implications for both empirical work and policy design. If the distribution of priors can be recovered from anonymous belief data, then policymakers can study belief heterogeneity without tracking the same individuals across multiple questions. This would expand the usefulness of cross-sectional surveys, anonymized experiments, and other forms of aggregate belief data for measuring disagreement in the population. If such recovery is impossible, the limitation is equally informative: it shows that aggregate belief distributions alone are insufficient, and that identification requires additional structure, richer data, or stronger assumptions.

%The goal of this paper is to identify exactly when aggregate belief data determine the underlying distribution of priors. The paper studies the map from a population distribution of priors to the collection of observable belief distributions generated by a family of events. It asks when we can uniquely pinned down the distributions of priors. In this way, the paper connects the familiar individual-level idea of identifying beliefs through relative likelihoods to a population-level problem where the analyst observes only anonymous distributions of beliefs.

Formally, the analyst observes the following object. For each event in some family (the primitive of the model), the analyst observes the distribution of posterior beliefs about the event across the population.  Primarily we imagine a fixed and known population size.  What the analyst does not observe is the identity of the agents behind these reports. Thus, the data reveal how beliefs are distributed event by event, but they do not reveal how a particular individual's beliefs about different events are linked. The main question is whether this anonymous collection of belief distributions is enough to recover the underlying distribution of priors in the population.

Our first result establishes a sharp limitation. If the analyst does not observe beliefs about the full state space, then identification can fail. There always exist distinct distributions of priors that generate the same observable data. The strength of this result is worth emphasizing: non-identification can arise even when the analyst observes posterior-belief distributions for an extremely rich collection of events, including all proper subsets of the entire state space. %Thus, observing ``almost'' the entire state space is not necessarily a substitute for observing the full state space itself when we identify the beliefs of Bayesian population.

The intuition is that anonymous data reveal only marginal distributions of posterior beliefs across events. Without observing beliefs about the full state space, the analyst may be unable to pin down how the different pieces of the belief distribution fit together. Even if every large subset of states is observed, the missing full event can leave enough freedom to construct two different population distributions of priors that agree on every observed event but differ globally. This result highlights a fundamental contrast between individual and population-level identification.

% Given \(\Sigma\), define a graph on the state space \(X\) by connecting two states \(x\) and \(y\) whenever they appear together in some observed event. Formally, we say that two states \(x\) and \(y\) are \emph{relatable}, written $x\sim_{\Sigma} y,$ if there exists an event $E\in\Sigma$ including both states, i.e.,
% $\{x,y\}\subseteq E$. This relation induces a graph on the state space \(X\): two states are connected whenever they appear together in some observed event. The graph therefore records which pairs of states are jointly compared in the data. Our analysis shows that identification depends critically on the structure of this graph. In particular, the key distinction is whether the graph is separable or nonseparable. Roughly speaking, a graph is separable if it can be divided into two smaller subgraphs that share exactly one state. 

While our first result establishes the existence of non-identification, this does not tell us whether non-identification is widespread or exceptional. The next results address precisely this question. Contrary to the first result, we provide a complementary positive result for $n$-agent distributions.\footnote{An $n$-agent distribution is one that could arise as the empirical
frequency of exactly $n$ agents. Specifically, every atom of $n$-agent distribution
must have a mass that is an integer multiple of $\tfrac{1}{n}$.} We show that the answer is governed by a simple graph-theoretic structure induced by the family of observed events.  We record the pairs of states whose probability ratios are observed in the data---any pair belonging to any elicited event. Our analysis shows that identification depends critically on whether the graph of this relation is separable or nonseparable. Roughly, a graph is separable if it can be divided into two smaller subgraphs that share exactly one state.  Unlike the single-agent case, connectedness is not enough.

When the graph induced by the observed data is nonseparable, identification is \emph{generic} for a given cardinality of agents: the set of identified distributions is open and dense. In this sense, non-identification is exceptional rather than typical. %A generic set is large from a topological point of view: it contains an open neighborhood around each of its points and comes arbitrarily close to every point in the grand space. 
Thus, as long as the induced graph is nonseparable, a sufficiently rich survey design makes non-uniqueness a non-generic concern. Put differently, when the family of observed events is ``rich enough'', one should not typically expect distinct distributions of priors to generate the same observable data. 

%{\color{red}CHRIS: SHOULD FOLLOWING PARAGRAPH GO IN MAIN TEXT?  IT'S NOT CLEAR THAT IT HELPS HERE, IN THE MAIN TEXT WE CAN DRAW A PICTURE OF THE CYCLE.  IT SHOULD BE OBVIOUS THE MINIMAL NONSEPARABLE GRAPH INDUCED BY BINARIES IS THE CYCLE, WHICH ONLY ADDS ONE MORE EVENT TO THE MINIMAL SIZE GRAPH OF BINARIES GIVING IDENTIFICATION FOR ONE PERSON.}

%It is important to note, however, that generic identification does not require eliciting large events. For instance, observing only binary events of the form $E_i=\{x_i,x_{i+1}\}$,  $i=1,\ldots,n-1$ together with $E_n=\{x_n,x_1\},$ is already sufficient for the relevant graph-theoretic condition to hold. These events generate a cycle on the state space. Hence, even though each elicited event contains only two states, the induced graph is nonseparable. This illustrates that what matters is not the size of each elicited event, but the way the elicited events connect the state space. A carefully designed collection of binary questions may therefore provide more identifying power than a seemingly richer collection of large events with the wrong graph structure.

%{\color{blue}CHRIS: AGAIN THE FOLLOWING RESULT USES NOTATION.  LET'S TALK ABOUT IT.}

    Our next result shows that the design of the elicited events is crucial. In contrast to the nonseparable case, if the induced graph separable, then non-identification is generic. %In the separable case, non-identification is large: it persists under small perturbations of the underlying distribution and arises arbitrarily close to every distribution in the grand space. This result highlights the importance of survey design. I
    It is not enough for the analyst to elicit beliefs about many events; the events must also connect the state space in the right way. %By contrast, designing \(\Sigma\) so that the induced graph is nonseparable eliminates this generic source of non-identification. 
    We then extend the analysis to the class of simple (finitely-supported) distributions. Enlarging the domain refines the picture in an important way:  non-separability of the induced graph no longer leads to generic identification.  In fact, both the set of beliefs for which identification prevails and the set of beliefs for which non-identification prevails are dense.  If the graph is separable, the set of non-identified distributions remains open and dense. % If $(X,\sim_\Sigma)$ is non-separable and $X \notin \Sigma$, both the set of identified and the set of non-identified distributions are dense --- neither is topologically negligible. If $(X,\sim_\Sigma)$ is separable, the set of non-identified distributions is open and dense.

The results show that identification from anonymous aggregate belief data is governed by the graph structure induced by family of observed events. The number or size of elicited events is not the primitive determinant of identification.  What matters is whether the events connect the state space in a way that prevents separation. When the induced graph is nonseparable, identification is generic for finite-type populations. When it is separable, non-identification is generic. Thus, the paper provides both a limitation and a design principle: anonymous belief data can be informative about population heterogeneity, but only when the elicited events are organized so as to connect the state space in the right way.

\section{The model}
Let $X$ denote the finite set of \emph{states}, and let $n\in\mathbb{N}$ denote the number of \emph{agents}.  We assume that $|X|\geq 3$, otherwise the questions asked in this paper become uninteresting.  $X$ is given the discrete topology and $\Delta(X)$ its weak topology (here simply the Euclidean topology on the simplex).

A simple probability $\pi$ on $\Delta(X)$ is an \emph{$n$-agent distribution} if for all $p\in\Delta(X)$, $\pi(\{p\})\in\{0,\tfrac{1}{n},\ldots,\tfrac{n-1}{n},1\}$.  That is, an $n$-agent distribution is one that could arise as the empirical frequency of exactly $n$ agents.  Denote the set of $n$-agent distributions by $\Delta^n(\Delta(X))$.  More generally, denote the set of finitely-supported (simple) distributions on $\Delta(X)$ by $\Delta^S(\Delta(X))$.  Observe that $\Delta^n(\Delta(X))\subseteq \Delta^S(\Delta(X))$.  Similarly define $\Delta^n(\Delta_{++}(X))$ and $\Delta^S(\Delta_{++}(X))$. The following example illustrates the notation. 
\begin{example}\label{ex:simplex} Let $X = \{1, 2, 3\}$, so $\Delta(X)$ is the set of beliefs over three states. A point $p \in \Delta(X)$ takes the form $(p_1, p_2, p_3)$ with $p_1 + p_2 + p_3 = 1$. A distribution $\pi \in \Delta^3(\Delta(X))$ assigns probability mass to such beliefs and can be interpreted as describing a population of agents, each holding some belief over $X$. For example, 
\[
  \pi\bigl(\{(\tfrac{1}{2}, \tfrac{1}{3}, \tfrac{1}{6})\}\bigr) = \tfrac{2}{3},
  \qquad
  \pi\bigl(\bigl\{\bigl(\tfrac{1}{4}, \tfrac{1}{4}, \tfrac{1}{2}\bigr)\bigr\}\bigr) = \tfrac{1}{3}.
\]
Both masses are multiples of $\tfrac{1}{3}$, so $\pi \in \Delta^3(\Delta(X))$. Now consider
\[
  \pi'\bigl(\{(1, 0, 0)\}\bigr) = \tfrac{1}{2},
  \qquad
  \pi'\bigl(\{(0, 1, 0)\}\bigr) = \tfrac{1}{2}.
\] Since $\tfrac{1}{2}$ is not a multiple of $\tfrac{1}{3}$, no assignment of exactly $3$ agents
to types could produce this distribution, so $\pi' \notin \Delta^3(\Delta(X))$. It is, however, a $2$-agent distribution: $\pi' \in \Delta^2(\Delta(X))$ (but not in $\Delta^2(\Delta_{++}(X))$). Both distributions above are finitely supported; hence, both belong to $\Delta^S(\Delta(X))$. The set $\Delta^S(\Delta(X))$ does not restrict the size of the atoms, whereas $\Delta^n(\Delta(X))$ requires every atom to be a multiple of $\tfrac{1}{n}$.
\end{example}

Recall that the \emph{weak topology} on $\Delta(\Delta(X))$ is the smallest topology for which, for each bounded continuous $f:\Delta(X)\rightarrow\mathbb{R}$, the map $\pi\mapsto \int f(p)\,d\pi(p)$ is continuous.  We equip $\Delta^n(\Delta_{++}(X))$ with the relative weak topology.

Next, let $\Sigma\subseteq 2^X\setminus\{\varnothing\}$ be a family of events. For $p\in\Delta_{++}(X)$ and $E\in\Sigma$, write $p(\cdot|E)\in\Delta_{++}(E)$ for the Bayesian update: $p(F|E)=p(F)/p(E)$ for $F\subseteq E$.  We say that $\pi\in\Delta^S(\Delta_{++}(X))$ is \emph{identified by $\Sigma$} if, whenever $\pi'\in\Delta^S(\Delta_{++}(X))$ satisfies: for every $E\in\Sigma$ and every finite set $P\subseteq \Delta_{++}(X)$,
\[
  \pi'\bigl(\{p\in\Delta_{++}(X):p(\cdot|E)\in P\}\bigr)
  = \pi\bigl(\{p\in\Delta_{++}(X):p(\cdot|E)\in P\}\bigr),
\]
then $\pi'=\pi$.  Otherwise $\pi$ is \emph{not identified by $\Sigma$}.

Given $\Sigma$, define an undirected graph $(X,\sim)$ on vertex set $X$ by setting $x \sim y$ (suppressing the dependence on $\Sigma$) if there exists $E\in\Sigma$ with $\{x,y\}\subseteq E$. The graph records only which pairs of states coexisting in some event, it disregards both the size and number of events that generate it. Many distinct families of events therefore induce the same graph. Figure~\ref{fig:induced} illustrates the construction. Adding a subset of an event present in $\Sigma$ does not change the graph, whereas adding a new event that nontrivially intersects two already present events can add new edges.

%%====================================================================
%%  FIGURE: how a family of events induces the graph, and how observed
%%  sets do or do not change its non-separable structure.
%%
%%  DROP-IN (Section 2).  Assumes the main preamble already provides:
%%    - \usetikzlibrary{backgrounds, calc, arrows.meta}
%%  This version defines a compact LOCAL layout (does not use \vpos),
%%  is scaled down to sit as a normal in-text figure, places the three
%%  panels in one row, and uses a SINGLE caption.
%%====================================================================

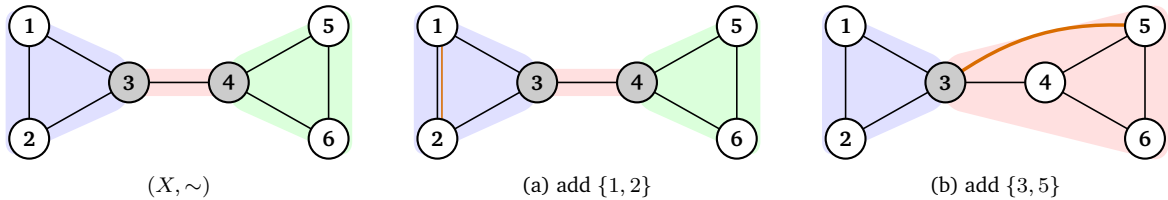
\begin{figure}[h!]
\centering

\tikzset{
  state/.style  = {circle, draw=black, fill=white,
                   minimum size=15pt, inner sep=0pt,
                   font=\scriptsize\bfseries, thick},
  cutvtx/.style = {circle, draw=black, fill=gray!40,
                   minimum size=15pt, inner sep=0pt,
                   font=\scriptsize\bfseries, thick},
  edge/.style   = {black, semithick},
  newedge/.style= {orange!85!black, line width=1.3pt},
  compA/.style  = {fill=blue!12,  rounded corners=7pt},
  compB/.style  = {fill=red!12,   rounded corners=7pt},
  compC/.style  = {fill=green!14, rounded corners=7pt},
  compBIG/.style= {fill=red!12,   rounded corners=7pt}
}

%% compact six-vertex layout, local to this figure
\newcommand{\vlayout}{%
  \coordinate (v1) at (-3.0, 0.62);
  \coordinate (v2) at (-3.0,-0.62);
  \coordinate (v3) at (-1.9, 0.0);
  \coordinate (v4) at (-0.8, 0.0);
  \coordinate (v5) at ( 0.3, 0.62);
  \coordinate (v6) at ( 0.3,-0.62);
}

\begin{tikzpicture}[scale=1.2]
%% ---------- (i) base graph ----------
\begin{scope}[shift={(0,0)}]
  \vlayout
  \begin{scope}[on background layer]
    \fill[compA] ($(v1)+(-0.26,0.24)$) -- ($(v2)+(-0.26,-0.24)$) --
                 ($(v3)+(0.04,-0.23)$) -- ($(v3)+(0.04,0.23)$) -- cycle;
    \fill[compB] ($(v3)+(-0.02,0.15)$) -- ($(v4)+(0.02,0.15)$) --
                 ($(v4)+(0.02,-0.15)$) -- ($(v3)+(-0.02,-0.15)$) -- cycle;
    \fill[compC] ($(v4)+(-0.04,0.23)$) -- ($(v5)+(0.26,0.24)$) --
                 ($(v6)+(0.26,-0.24)$) -- ($(v4)+(-0.04,-0.23)$) -- cycle;
  \end{scope}
  \draw[edge] (v1)--(v2); \draw[edge] (v1)--(v3); \draw[edge] (v2)--(v3);
  \draw[edge] (v3)--(v4);
  \draw[edge] (v4)--(v5); \draw[edge] (v4)--(v6); \draw[edge] (v5)--(v6);
  \node[state]  at (v1) {1}; \node[state]  at (v2) {2};
  \node[cutvtx] at (v3) {3}; \node[cutvtx] at (v4) {4};
  \node[state]  at (v5) {5}; \node[state]  at (v6) {6};
  \node[font=\scriptsize] at (-1.35,-1.15) {$(X,\sim)$};
\end{scope}
%% ---------- (a) add subset ----------
\begin{scope}[shift={(4.5,0)}]
  \vlayout
  \begin{scope}[on background layer]
    \fill[compA] ($(v1)+(-0.26,0.24)$) -- ($(v2)+(-0.26,-0.24)$) --
                 ($(v3)+(0.04,-0.23)$) -- ($(v3)+(0.04,0.23)$) -- cycle;
    \fill[compB] ($(v3)+(-0.02,0.15)$) -- ($(v4)+(0.02,0.15)$) --
                 ($(v4)+(0.02,-0.15)$) -- ($(v3)+(-0.02,-0.15)$) -- cycle;
    \fill[compC] ($(v4)+(-0.04,0.23)$) -- ($(v5)+(0.26,0.24)$) --
                 ($(v6)+(0.26,-0.24)$) -- ($(v4)+(-0.04,-0.23)$) -- cycle;
  \end{scope}
  \draw[edge] (v1)--(v2); \draw[edge] (v1)--(v3); \draw[edge] (v2)--(v3);
  \draw[edge] (v3)--(v4);
  \draw[edge] (v4)--(v5); \draw[edge] (v4)--(v6); \draw[edge] (v5)--(v6);
  \draw[orange!85!black, line width=0.6pt]
       ($(v1)+(0.05,0)$)--($(v2)+(0.05,0)$);
  \node[state]  at (v1) {1}; \node[state]  at (v2) {2};
  \node[cutvtx] at (v3) {3}; \node[cutvtx] at (v4) {4};
  \node[state]  at (v5) {5}; \node[state]  at (v6) {6};
  \node[font=\scriptsize] at (-1.35,-1.15) {(a) add $\{1,2\}$};
\end{scope}
%% ---------- (b) add crossing set ----------
\begin{scope}[shift={(9.0,0)}]
  \vlayout
  \begin{scope}[on background layer]
    \fill[compA] ($(v1)+(-0.26,0.24)$) -- ($(v2)+(-0.26,-0.24)$) --
                 ($(v3)+(0.04,-0.23)$) -- ($(v3)+(0.04,0.23)$) -- cycle;
    \fill[compBIG] ($(v3)+(-0.04,0.27)$) -- ($(v5)+(0.26,0.27)$) --
                   ($(v6)+(0.26,-0.27)$) -- ($(v3)+(-0.04,-0.27)$) -- cycle;
  \end{scope}
  \draw[edge] (v1)--(v2); \draw[edge] (v1)--(v3); \draw[edge] (v2)--(v3);
  \draw[edge] (v3)--(v4);
  \draw[edge] (v4)--(v5); \draw[edge] (v4)--(v6); \draw[edge] (v5)--(v6);
  \draw[newedge] (v3) to[bend left=20] (v5);
  \node[state]  at (v1) {1}; \node[state]  at (v2) {2};
  \node[cutvtx] at (v3) {3}; \node[state]  at (v4) {4};
  \node[state]  at (v5) {5}; \node[state]  at (v6) {6};
  \node[font=\scriptsize] at (-1.35,-1.15) {(b) add $\{3,5\}$};
\end{scope}
\end{tikzpicture}

\caption{\footnotesize Two states are joined whenever some event of $\Sigma$
contains both, so the induced graph $(X,\sim)$ records only pairwise
co-occurrences. Here
$X=\{1,2,3,4,5,6\}$ and $\Sigma=\{\{1,2,3\},\{3,4\},\{4,5,6\}\}$ yields a triangle, an edge, and a
triangle: three maximal non-separable subgraphs (shaded), joined at states
$3$ and $4$, each of whose removal disconnects the graph. Adding a subset of
an existing event (a) induces no new edge and leaves the subgraphs unchanged,
whereas adding a set whose elements cross two subgraphs (b) induces the edge
$\{3,5\}$ that merges them, so that $4$ no longer disconnects the graph.}
\label{fig:induced}

\end{figure}

The following result captures identification for a single Bayesian agent. Since it is standard, we state it without proof.\footnote{To generalize beyond the case of $\Delta_{++}(X)$ to $\Delta(X)$, we would need to assume that $(X,\sim)$ is the complete graph.}

\begin{proposition}\label{prop:connected}The following are equivalent:
\begin{enumerate}
\item $(X,\sim)$ is a connected graph.
\item For every $p,q\in\Delta_{++}(X)$, if $p(\cdot|E)=q(\cdot|E)$ for all $E\in\Sigma$, then $p=q$.
\end{enumerate}
\end{proposition}

Proposition~\ref{prop:connected} motivates our study. For a single Bayesian agent, a probability measure over a finite set is uniquely pinned down---that is, it is \emph{identified}---by its relative likelihoods.  In fact, according to the proposition, we can identify any such probability distribution from its relative likelihoods on a restricted domain of pairs---generally for a set $X$, we only need $|X|-1$ such pairs to identify $p$. For example, when $X=\{1,2,3\}$, observing only the conditional distributions associated with $\Sigma=\bigl\{\{1,3\},\{2,3\}\bigr\}$ is sufficient to identify the belief over the full state space for a single Bayesian agent. Henceforth, for simplicity, we make the following assumption:\footnote{In case $(X,\sim)$ is not connected, non-identification is universal for full support profiles.}

\textbf{Assumption:}  $(X,\sim)$ is a connected graph.

%we assume $(X,\sim)$ is a connected graph. 
The following example shows Proposition~\ref{prop:connected} does not hold for a population of Bayesian agents.

\begin{example}\label{ex:non-identification}  Let $\Sigma=\{ \{1,3\},\{2,3\}\}$ and $p^1=(\tfrac{1}{2}, \tfrac{1}{3}, \tfrac{1}{6})$ and $p^2=\bigl(\tfrac{1}{4}, \tfrac{1}{4}, \tfrac{1}{2}\bigr)$, and 
$\pi\bigl(\{p^1\}\bigr) = \tfrac{2}{3}$ and $ \pi\bigl(\{p^2 \}\bigr) = \tfrac{1}{3}$ as in Example \ref{ex:simplex}.
We now show that $\pi$ is \emph{not identified by $\Sigma$}. First consider $q^1=p^1$, $q^2=(\tfrac{6}{9}, \tfrac{1}{9}, \tfrac{2}{9})$ and $q^3=\bigl(\tfrac{1}{7}, \tfrac{4}{7}, \tfrac{2}{7}\bigr)$. Define $\pi'$ such that $\pi'\bigl(\{q^i\}\bigr)= \tfrac{1}{3}$.
  Under $E=\{1,3\}$, $p^1(\cdot|E)=q^1(\cdot|E)=q^2(\cdot|E)$, and $p^2(\cdot|E)=q^3(\cdot|E)$. Similarly, under $E'=\{2,3\}$, $p^1(\cdot|E')=q^1(\cdot|E')=q^3(\cdot|E')$, and $p^2(\cdot|E')=q^2(\cdot|E')$ (see Figure \ref{fig:simplex-nonidentification}). Hence, while both $\pi $ and $ \pi'$ induce the same distribution on $\Sigma$, they are distinct, $\pi \neq \pi'$. 
\end{example}

\begin{figure}[h!]
    \centering
    \begin{tikzpicture}[scale=2, >=stealth]

%% ---------------------------------------------------------------
%% Coordinates
%% Equilateral triangle, side 4; barycentric → Cartesian:
%%   x = 2(p_3 - p_2),  y = 2√3 · p_1
%% ---------------------------------------------------------------

% Vertices
\coordinate (V1) at ( 0,     3.464);  % state 1 (top)
\coordinate (V2) at (-2,     0);      % state 2 (bottom-left)
\coordinate (V3) at ( 2,     0);      % state 3 (bottom-right)

% Lottery points
\coordinate (P1)   at (-0.333, 1.732); % p^1 = q^1 = (1/2, 1/3, 1/6)
\coordinate (P2)   at ( 0.5,   0.866); % p^2       = (1/4, 1/4, 1/2)
\coordinate (Q2pt) at ( 0.222, 2.309); % q^2       = (2/3, 1/9, 2/9)
\coordinate (Q3pt) at (-0.571, 0.495); % q^3       = (1/7, 4/7, 2/7)

% Iso-conditional endpoints on event edges
% E = {1,3}  →  edge V1–V3  (p_2 = 0)
\coordinate (EA1) at ( 0.5,   2.598); % p(·|E) = (3/4, 1/4)
\coordinate (EA2) at ( 1.333, 1.155); % p(·|E) = (1/3, 2/3)
% E'= {2,3}  →  edge V2–V3  (p_1 = 0)
\coordinate (EB1) at (-0.667, 0);     % p(·|E')= (2/3, 1/3)
\coordinate (EB2) at ( 0.667, 0);     % p(·|E')= (1/3, 2/3)

%% ---------------------------------------------------------------
%% Iso-conditional lines (drawn first, behind simplex boundary)
%% ---------------------------------------------------------------

% E = {1,3}: rays from V2  (blue, solid)
\draw[cyan!50, thick, dashed] (V2) -- (EA1);
\draw[cyan!50, thick, dashed] (V2) -- (EA2);

% E'= {2,3}: rays from V1  (red, dashed)
\draw[red!55, thick, dashed] (V1) -- (EB1);
\draw[red!55, thick, dashed] (V1) -- (EB2);

%% ---------------------------------------------------------------
%% Simplex boundary
%% ---------------------------------------------------------------
\draw[very thick] (V1) -- (V2) -- (V3) -- cycle;

% % Highlight event edges
% \draw[blue!75, line width=2.2pt] (V1) -- (V3); % E  = {1,3}
% \draw[red!75,  line width=2.2pt] (V2) -- (V3); % E' = {2,3}

%% ---------------------------------------------------------------
%% Vertex labels
%% ---------------------------------------------------------------
\node[above,       font=\small\bfseries] at (V1) {$1$};
\node[below left,  font=\small\bfseries] at (V2) {$2$};
\node[below right, font=\small\bfseries] at (V3) {$3$};

%% ---------------------------------------------------------------
%% Event-edge labels
%% ---------------------------------------------------------------
\node[font=\scriptsize, cyan!80, rotate=-60] at (1.55, 2.15)
    {$E=\{1,3\}$};
\node[font=\scriptsize, red!80] at (0, -0.35)
    {$E'=\{2,3\}$};

%% ---------------------------------------------------------------
%% Conditional-value labels at ray endpoints on edges
%% ---------------------------------------------------------------
% \filldraw[blue!55] (EA1) circle (1.5pt);
% \node[right=2pt, font=\tiny, blue!65] at (EA1)
%     {$\bigl(\tfrac{3}{4},\tfrac{1}{4}\bigr)$};

% \filldraw[blue!55] (EA2) circle (1.5pt);
% \node[right=2pt, font=\tiny, blue!65] at (EA2)
%     {$\bigl(\tfrac{1}{3},\tfrac{2}{3}\bigr)$};

% \filldraw[red!55] (EB1) circle (1.5pt);
% \node[below=2pt, font=\tiny, red!65] at (EB1)
%     {$\bigl(\tfrac{2}{3},\tfrac{1}{3}\bigr)$};

% \filldraw[red!55] (EB2) circle (1.5pt);
% \node[below=2pt, font=\tiny, red!65] at (EB2)
%     {$\bigl(\tfrac{1}{3},\tfrac{2}{3}\bigr)$};

%% ---------------------------------------------------------------
%% Atoms of π  (blue filled discs; larger = more weight)
%% ---------------------------------------------------------------
\filldraw[fill=blue!25, draw=blue!80, thick] (P1) circle (2pt); % w=2/3
\filldraw[fill=blue!25, draw=blue!80, thick] (P2) circle (1pt); % w=1/3

%% ---------------------------------------------------------------
%% Atoms of π' (green rings; q^1 = p^1 drawn as outer ring)
%% ---------------------------------------------------------------
\draw[fill=green!50!black, green!50!black,  thick] (P1) circle (1pt);            % q^1, w=1/3
\filldraw[fill=green!50!black, draw=green!50!black, thick]
    (Q2pt) circle (1.0pt);                                         % q^2, w=1/3
\filldraw[fill=green!50!black, draw=green!50!black, thick]
    (Q3pt) circle (1.0pt);                                         % q^3, w=1/3

%% ---------------------------------------------------------------
%% Point labels
%% ---------------------------------------------------------------
\node[left=3pt,  font=\scriptsize, green!50!black] at (P1)   {$q^1$};
\node[right=4pt,  font=\scriptsize, blue!80] at (P1)   {$p^1$};
\node[right=2pt,  font=\scriptsize, blue!80] at (P2)   {$p^2$};
\node[left=2pt,  font=\scriptsize, green!50!black] at (Q2pt) {$q^2$};
\node[left=2pt,   font=\scriptsize, green!50!black] at (Q3pt) {$q^3$};

%% ---------------------------------------------------------------
%% Legend
%% ---------------------------------------------------------------
% \begin{scope}[shift={(-2.7, -0.85)}, font=\tiny]
%     \draw[blue!55, thick]
%         (0, 0.50) -- (0.45, 0.50)
%         node[right] {iso-cond.\ under $E=\{1,3\}$};
%     \draw[red!55, thick, dashed]
%         (0, 0.22) -- (0.45, 0.22)
%         node[right] {iso-cond.\ under $E'=\{2,3\}$};
%     \filldraw[fill=blue!25, draw=blue!80, thick]
%         (0.22, -0.08) circle (2.8pt)
%         node[right=4pt] {atom of $\pi$ \ (size $\propto$ weight)};
%     \draw[green!50!black, thick]
%         (0.22, -0.36) circle (2.8pt)
%         node[right=4pt] {atom of $\pi'$};
% \end{scope}

\end{tikzpicture}
\caption{\footnotesize{Illustration of non-identification in Example \ref{ex:non-identification}. Each point in the simplex represents a belief over states of the world: $\{1,2,3\}$. Blue discs are the atoms of $\pi$ (area proportional to weight), and green discs are the atoms of $\pi'$.  The point $p^1=q^1$ carries weight $\tfrac{2}{3}$ under $\pi$ and $\tfrac{1}{3}$ under $\pi'$. Dashed lines are iso-conditional rays for $E=\{1,3\}$ and $E'=\{2,3\}$. All beliefs on a given ray share the same conditional. $\{p^1\}$ and $\{q^1,q^2\}$ receive total
weight $\tfrac{2}{3}$ under both $\pi$ and $\pi'$, and likewise for the
remaining rays.  Hence $\pi$ and $\pi'$ induce identical distributions over
conditionals for every $E\in\Sigma=\bigl\{\{1,3\},\{2,3\}\bigr\}$, so
$\pi$ is not identified by~$\Sigma$.}}
\label{fig:simplex-nonidentification}
\end{figure}
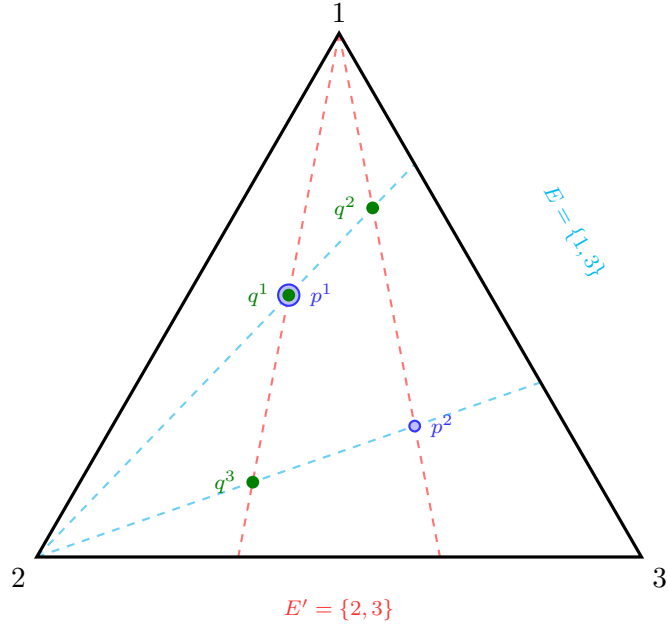

 The example shows that the identification argument does not extend directly to a population of Bayesian agents.  The difficulty arises when the iso-conditional rays associated with different events align across distinct atoms in the same way. In that case, one can construct an alternative distribution \(\pi'\) by splitting the atoms of \(\pi\) and recombining their masses along these rays. This transformation preserves every observable conditional distribution, while changing the underlying distribution over full beliefs.

Thus, the failure of identification is not a knife-edge coincidence. Rather, it reflects a fundamental limitation of \(\Sigma\) as an identifying device: whenever the events in \(\Sigma\) do not jointly separate all pairs of beliefs in the support, the distribution over full-state beliefs cannot, in general, be recovered from the observed conditionals alone.

Additional information can improve identification. For example, adding $E''=\{1,2\}$ to $\Sigma=\{ \{1,3\},\{2,3\}\}$ improves the identification. This event introduces a third family of iso-conditional rays, emanating from vertex~$3$ (the state outside $E''$), which cross-cuts the existing groupings.  Under $E''$, one computes $p^1(\cdot|\{1,2\})=(\tfrac{3}{5},\tfrac{2}{5},0)$ while $q^2(\cdot|\{1,2\})=(\tfrac{6}{7},\tfrac{1}{7},0)$, so $p^1$ and $q^2$---which lie on the same ray from vertex~$2$ under $E=\{1,3\}$---are now separated.  Similarly, $p^2(\cdot|\{1,2\})=(\tfrac{1}{2},\tfrac{1}{2},0)$ and $q^3(\cdot|\{1,2\})=(\tfrac{1}{5},\tfrac{4}{5},0)$ differ, separating $p^2$ and $q^3$ as well.  The non-identification failure exploited the fact that mass could be split and recombined along the rays of $\{1,3\}$ and $\{2,3\}$ without disturbing either observable; the rays of $\{1,2\}$ block this maneuver by imposing additional constraints that no alternative $\pi'\neq\pi$ can simultaneously satisfy.  This illustrates a broader principle: the critical question is whether the events in $\Sigma$ jointly separate all pairs of lotteries in the support of~$\pi$.

\section{Non-identification}

In this section, we first illustrate  that, regardless of how rich $\Sigma$
 is, the mere absence of the grand menu $X$ from $\Sigma$ is sufficient to
 guarantee the existence of a non-identified distribution.  The example
 below illustrates this for $|X|=3$.

% We have shown that if $\sim_{\Sigma}$ is nonseparable, then generic identification obtains.  Still, identification is not global.  Let us refer to $\Delta^{n,u}(\Delta_{++}(X))=\{\pi\in\Delta(\Delta_{++}(X)):\mbox{ for all }p,\pi(p)\in\{0,\frac{1}{n}\}\}$.  This set is the image of $\Delta^u(X)^n$ under $\psi$. The next example illustrates that there is always a non-identified $\pi$ independent of how rich $\Sigma$ as long as $X$ does not belongs to $\Sigma$.

% Theorem~\ref{thm:separable} establishes that generic identification holds
% whenever $\sim_{\Sigma}$ is nonseparable.  Generic, however, does not mean
% universal: non-identified distributions always exist.  To state the next
% result precisely, let
% \[
%   \Delta^{n,u}(\Delta_{++}(X))
%   \;=\;
%   \bigl\{\pi\in\Delta(\Delta_{++}(X)):\pi(p)\in\{0,\tfrac{1}{n}\}
%   \text{ for all }p\bigr\},
% \]
% the image of $\Delta^u(X)^n$ under $\psi$, consisting of the
% $n$-agent distributions that are uniform over their support.
% Theorem~\ref{thm:nonidentify} shows that, regardless of how rich $\Sigma$
% is, the mere absence of the grand menu $X$ from $\Sigma$ is sufficient to
% guarantee the existence of a non-identified distribution.  The example
% below illustrates the construction for $|X|=3$.

\begin{example}
Let $X=\{1,2,3\}$ and suppose $X\notin\Sigma$.  Define
\begin{align*}
   p^{\{1,2\}} &= \tfrac{1}{5}(2,2,1), &
  p^{\{1,3\}} &= \tfrac{1}{5}(2,1,2), &
  p^{\{2,3\}} &= \tfrac{1}{5}(1,2,2),\\[4pt]
  p^{\{1\}}   &= \tfrac{1}{4}(2,1,1), &
  p^{\{2\}}  &= \tfrac{1}{4}(1,2,1), &
  p^{\{3\}}  &= \tfrac{1}{4}(1,1,2).
\end{align*} 

Note that $p^E$ places equal extra weight on every element of $E$ relative to the uniform distribution. Define $\pi$ and $\pi'$ are two uniform distribution on $\{p^{\{1,2\}},p^{\{2,3\}},p^{\{1,3\}}\}$ and $\{p^{\{1\}},p^{\{2\}},p^{\{3\}}\}$, respectively. The two distributions have disjoint supports and are clearly distinct; each
are uniform distributions over three probability measures.%.{\color{red}CHRIS: THIS NOTATION WAS MOVED INTO A PROOF, WE SHOULD EITHER MOVE IT OUTSIDE OF THE PROOF OR NOT USE IT HERE.  WE NEED TO BE CAREFUL WITH CLAUDE BECAUSE IT MOVES NOTATION INTO PROOFS}

Take $E=\{1,2\}\in\Sigma$.  Computing: $p^{\{1,2\}}(\cdot|E)=p^{\{3\}}(\cdot|E)$, $p^{\{1,3\}}(\cdot|E)=p^{\{1\}}(\cdot|E)$, and $p^{\{2,3\}}(\cdot|E)=p^{\{2\}}(\cdot|E)$. Both distributions assign uniform weights, so the induced distributions coincide.  By the symmetry of the construction, the same
holds for $E=\{1,3\}$ and $E=\{2,3\}$, and hence for every
$E\in\Sigma$ with $E\subsetneq X$. It is easy to see (by definition) that identification will be restored by introducing observation the full menu. Under $E=X$, $p(1|X)=p_1$ takes values
$\tfrac{2}{5},\tfrac{2}{5},\tfrac{1}{5}$ under $\pi$ and
$\tfrac{1}{2},\tfrac{1}{4},\tfrac{1}{4}$ under $\pi'$.  Hence adding $X$ to $\Sigma$ identifies $\pi$.
\end{example}

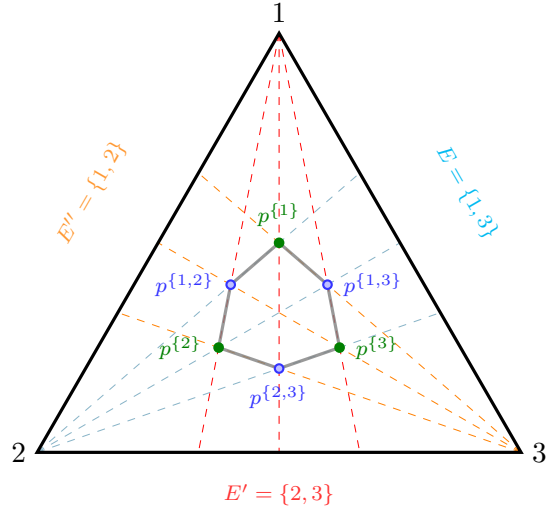
\begin{figure}[h!]
\centering
\begin{tikzpicture}[scale=1.6, >=stealth]

%% ---------------------------------------------------------------
%%  Coordinates.  Barycentric -> Cartesian (side = 4):
%%    X = 2(p_3 - p_2),   Y = 2sqrt(3)*p_1
%%  Vertex 1 at top, 2 at bottom-left, 3 at bottom-right.
%% ---------------------------------------------------------------
\coordinate (Vx) at ( 0,      3.464);
\coordinate (Vy) at (-2,      0);
\coordinate (Vz) at ( 2,      0);

%% pi atoms  (p_E = (1_X + 1_E)/(3+|E|), even subsets minus empty)
\coordinate (P1)  at (-0.400,  1.386);  % p^1 = (2,2,1)/5
\coordinate (P2)  at ( 0.400,  1.386);  % p^2 = (2,1,2)/5
\coordinate (P3)  at ( 0.000,  0.693);  % p^3 = (1,2,2)/5

%% pi' atoms  (odd singletons)
\coordinate (P1p) at ( 0.000,  1.732);  % q^1 = (2,1,1)/4
\coordinate (P2p) at (-0.500,  0.866);  % q^2 = (1,2,1)/4
\coordinate (P3p) at ( 0.500,  0.866);  % q^3 = (1,1,2)/4

%% Iso-conditional endpoints on edges
%% {1,2}-rays from Vz, on edge V1-V2  (p_3 = 0)
\coordinate (E12_A) at (-1.000, 1.732);  % p(1|{1,2}) = 1/2
\coordinate (E12_B) at (-0.667, 2.309);  % p(1|{1,2}) = 2/3
\coordinate (E12_C) at (-1.333, 1.155);  % p(1|{1,2}) = 1/3
%% {1,3}-rays from Vy, on edge V1-V3  (p_2 = 0)
\coordinate (E13_A) at ( 0.667, 2.309);  % p(1|{1,3}) = 2/3
\coordinate (E13_B) at ( 1.000, 1.732);  % p(1|{1,3}) = 1/2
\coordinate (E13_C) at ( 1.333, 1.155);  % p(1|{1,3}) = 1/3
%% {2,3}-rays from Vx, on edge V2-V3  (p_1 = 0)
\coordinate (E23_A) at (-0.667, 0);      % p(2|{2,3}) = 2/3
\coordinate (E23_B) at ( 0.667, 0);      % p(2|{2,3}) = 1/3
\coordinate (E23_C) at ( 0.000, 0);      % p(2|{2,3}) = 1/2

%% ---------------------------------------------------------------
%%  Iso-conditional lines (background)
%% ---------------------------------------------------------------
\draw[orange, ultra thin, dashed]            (Vz) -- (E12_A);
\draw[orange, ultra thin, dashed]            (Vz) -- (E12_B);
\draw[orange, ultra thin, dashed]            (Vz) -- (E12_C);
\draw[cyan!55!black!55, very thin, dashed]   (Vy) -- (E13_A);
\draw[cyan!55!black!55, very thin, dashed]   (Vy) -- (E13_B);
\draw[cyan!55!black!55, very thin, dashed]   (Vy) -- (E13_C);
\draw[red, ultra thin, dashed]               (Vx) -- (E23_A);
\draw[red, ultra thin, dashed]               (Vx) -- (E23_B);
\draw[red, ultra thin, dashed]               (Vx) -- (E23_C);

%% ---------------------------------------------------------------
%%  Pappus hexagon: q1-p1-q2-p3-q3-p2
%%  Sides pass through V2,V1,V3,V2,V1,V3 in order
%% ---------------------------------------------------------------
\draw[gray, very thick, opacity=0.8]
    (P1p) -- (P1) -- (P2p) -- (P3) -- (P3p) -- (P2) -- cycle;

%% ---------------------------------------------------------------
%%  Simplex boundary
%% ---------------------------------------------------------------
\draw[very thick] (Vx) -- (Vy) -- (Vz) -- cycle;

%% ---------------------------------------------------------------
%%  Vertex labels
%% ---------------------------------------------------------------
\node[above, font=\small\bfseries] at (Vx) {$1$};
\node[left,  font=\small\bfseries] at (Vy) {$2$};
\node[right, font=\small\bfseries] at (Vz) {$3$};

%% ---------------------------------------------------------------
%%  Event-edge labels
%% ---------------------------------------------------------------
\node[font=\scriptsize, cyan!80,   rotate=-60] at ( 1.55,  2.15)
    {$E=\{1,3\}$};
\node[font=\scriptsize, red!80]                at ( 0,    -0.35)
    {$E'=\{2,3\}$};
\node[font=\scriptsize, orange!80, rotate= 60] at (-1.55,  2.15)
    {$E''=\{1,2\}$};

%% ---------------------------------------------------------------
%%  Point labels
%% ---------------------------------------------------------------
\node[left=2pt,  font=\scriptsize, blue!80]       at (P1)  {$p^{\{1,2\}}$};
\node[right=2pt, font=\scriptsize, blue!80]       at (P2)  {$p^{\{1,3\}}$};
\node[below=2pt, font=\scriptsize, blue!80]       at (P3)  {$p^{\{2,3\}}$};
\node[above=2pt, font=\scriptsize, green!50!black] at (P1p) {$p^{\{1\}}$};
\node[left=2pt,  font=\scriptsize, green!50!black] at (P2p) {$p^{\{2\}}$};
\node[right=2pt, font=\scriptsize, green!50!black] at (P3p) {$p^{\{3\}}$};

%% ---------------------------------------------------------------
%%  pi atoms: blue discs
%% ---------------------------------------------------------------
\filldraw[fill=blue!25, draw=blue!80, thick] (P1) circle (1pt);
\filldraw[fill=blue!25, draw=blue!80, thick] (P2) circle (1pt);
\filldraw[fill=blue!25, draw=blue!80, thick] (P3) circle (1pt);

%% ---------------------------------------------------------------
%%  pi' atoms: green discs
%% ---------------------------------------------------------------
\filldraw[fill=green!50!black, draw=green!50!black, thick] (P1p) circle (1pt);
\filldraw[fill=green!50!black, draw=green!50!black, thick] (P2p) circle (1pt);
\filldraw[fill=green!50!black, draw=green!50!black, thick] (P3p) circle (1pt);

\end{tikzpicture}
\caption{\footnotesize{Symmetric non-identification in $\Delta(\{1,2,3\})$, following the construction of Theorem~\ref{thm:nonidentify}.  Blue discs are the atoms of
$\pi$ (pair-labeled distributions), and green discs are the atoms of
$\pi'$ (singleton-labeled distributions). The two distributions are indistinguishable on all
binary menus but are separated by the full menu.}}
\label{fig:thm-nonidentification}
\end{figure}

% \begin{theorem}\label{thm:nonidentify}Suppose that $X\notin \Sigma$.  
% \begin{itemize}
% \item If $|X|$ is even, there is $\pi\in\Delta^{2^{n-1},u}(\Delta_{++}(X))$ that is not identified.
% \item If $|X|$ is odd, there is $\pi\in\Delta^{2^{n-1}-1,u}(\Delta_{++}(X))$ that is not identified.
% \end{itemize}
% \end{theorem}

The next result states that a similar construction is possible for all $|X|$.

\begin{theorem}\label{thm:nonidentify}Suppose that $X\notin \Sigma$.  
There is a uniform population of Bayesian agents that is not identified by $\Sigma$. \end{theorem}

\begin{proof} Let  \[
  \Delta^{n,u}(\Delta_{++}(X))
  \;=\;
  \bigl\{\pi\in\Delta(\Delta_{++}(X)):\pi(p)\in\{0,\tfrac{1}{n}\}
  \text{ for all }p\bigr\},
\]
consist of the $n$-agent distributions that are uniform over their support. We prove this claim in two cases: (i) if $|X|$ is even, there is $\pi'\in\Delta^{2^{|X|-1},u}(\Delta_{++}(X))$ that is not identified, (ii) if $|X|$ is odd, there is $\pi\in\Delta^{2^{|X|-1}-1,u}(\Delta_{++}(X))$ that is not identified.

Let us consider the even case first.  For each $E\subseteq X$, we define $p_E=\frac{\mathbf{1}_X+\mathbf{1}_E}{|X|+|E|}$.  Now, set $\pi=\frac{1}{2^{|X|-1}}\sum_{E:|E|\mbox{ is even }}\delta_{p_E}$ and $\pi'=\frac{1}{2^{|X|-1}}\sum_{E:|E|\mbox{ is odd }} \delta_{p_E}$.  Except for the cases $\varnothing$ and $X$ (each of which have even cardinality), it is clear that if $E\neq F$, then $p_E\neq p_F$, so in fact $\pi$ and $\pi'$ have disjoint supports.

The cardinality of the family of even cardinality subsets and the cardinality of the family of odd subsets of a set with cardinality at least one are the same, a well-known consequence of the binomial formula:  $0=(1-1)^p=\sum_{k=0}^p\binom{p}{k}(-1)^k$, implying that $\sum_{\{k:k\mbox{ odd}\}}\binom{p}{k}=\sum_{\{k:k\mbox{ even}\}}\binom{p}{k}$.  In turn then each of these are $2^{p-1}$ (half the total number of subsets).

Now, fix any $L\subseteq X$ where $L\neq X$, so that $|X\setminus L|\geq 1$.  Then for any set $F\subseteq L$, $|\{E\subseteq X:|E|\mbox{ even and }E\cap L=F\}|=|\{E\subseteq X:|E|\mbox{ odd and }E\cap L = F\}|$, this follows as if $|F|$ is odd the first set corresponds to the cardinality of odd subsets of $X\setminus L$ and the second to the cardinality of even subsets of $X\setminus L$.  Similarly if $|F|$ is even the first set corresponds to the cardinality of even subsets of $X\setminus L$ and the second to the cardinality of odd subsets of $X\setminus L$.

Observe that for any $L\subseteq X$ and any $E\subseteq X$, $p_E(\cdot|L)=\frac{\mathbf{1}_L+\mathbf{1}_{E\cap L}}{|L|+|E\cap L|}$.  

First, consider the case in which $M\neq \varnothing$ and $M\neq L$.  Clearly if $M'\subseteq L$ is such that $\frac{\mathbf{1}_L+\mathbf{1}_M}{|L|+|M|}=\frac{\mathbf{1}_L+\mathbf{1}_{M'}}{|L|+|M'|}$, then $M=M'$.  So we can calculate $\pi(\{p_E:p_E(\cdot|L)=\frac{\mathbf{1}_L+\mathbf{1}_M}{|L|+|M|}\})$ as 
$\frac{1}{2^{|X|-1}}|\{E\subseteq X:E\mbox{ even },E\cap L=M\}|$, which we determined in the previous paragraph was equal to $\frac{1}{2^{|X|-1}}2^{|X|-|L|-1}$, or $\frac{1}{2^{|L|}}$, similarly for $\pi'$.  

Now, for the case in which $M\in\{\varnothing,L\}$, observe that $\frac{\mathbf{1}_L+\mathbf{1}_{\varnothing}}{|L|+|\varnothing|}=\frac{\mathbf{1}_L+\mathbf{1}_{L}}{|L|+|L|}$.  So to calculate $\pi(\{p_E:p_E(\cdot|L)=\frac{\mathbf{1}_L+\mathbf{1}_{\varnothing}}{|L|+|\varnothing|}\})$, we need to sum the number of $E\subseteq X$ for which $E\cap L=\varnothing$ to the number of $E\subseteq X$ for which $E\cap L=L$.  This results in, under either $\pi$ or $\pi'$, a probability of $\frac{1}{2^{|X|-1}}2^{|X|-|L|}$, or $\frac{1}{2^{|L|-1}}$.

%Then the number of even sets $E\subseteq X$ for which $E\cap L=M$ is, by the preceding, equal to the number of odd sets $E\subseteq X$ for which $E\cap L=M$, 

So, $\pi'$ is a distribution which is uniform (if $E$ and $F$ are odd and $p_E=p_F$, then $E=F$), though this is not the case for $\pi$.  And $\pi$ and $\pi'$ are distinct, having disjoint supports.  

The second case is similar, except that in this case, $p_{\varnothing}$ and $p_X$ correspond to the same element of $\Delta_{++}(X)$, and since one appears in the support of each of $\pi$ and $\pi'$, we may remove them with no overall effect on the induced probabilities of updates.

\end{proof}

%% MAIN THEOREM
\section{Identification on $\Delta^n(\Delta_{++}(X))$}
%%==========================================================================
In the last section, we establish the existence of non-identification. However, this does not tell us whether non-identification is widespread or exceptional. In this section, we address this question. We show that the answer is governed by the structure of $(X,\sim)$. %We record which pairs of states are jointly observed in the data. Our analysis shows that whether identification obtains generically depends on whether the graph is separable or nonseparable. 

We now borrow some definitions from \citet{bondy2008}.  A \emph{separation} of a graph is a decomposition into two nonempty connected subgraphs that share exactly one vertex in common.\footnote{A graph decomposition is a family of subgraphs whose edge sets partition the edges of the original graph, where a subgraph is any pair consisting of a subset $X'$ of vertices and an edge relation $\sim'$ satisfying: $x'\sim' y'$ implies $x' \sim y'$.}  The shared vertex $x^*$ is called a \emph{separating vertex}.  A connected graph that admits a separation is \emph{separable}; otherwise it is \emph{non-separable} (equivalently, \emph{2-connected}).

\begin{figure}[h!]
    \centering
\begin{tikzpicture}[
  vtx/.style    = {circle, draw=black, fill=white,
                   inner sep=0pt, minimum size=7pt, thick},
  sepvtx/.style = {circle, draw=black, fill=gray!40,
                   inner sep=0pt, minimum size=8pt, thick},
  thick
]

%%============================================================
%% LEFT: Non-separable (K_4, i.e.\ 2-connected)
%%============================================================

%% Vertices
\node[vtx] (v1) at (-2,  1) {};
\node[vtx] (v2) at (-2, -0.8) {};
\node[vtx] (v3) at (-3.8, -0.8) {};
\node[vtx] (v4) at (-3.8,  1) {};   % inner vertex

%% Edges: outer triangle + three spokes to inner vertex
\draw (v1) -- (v2) -- (v3) -- (v1);
\draw (v4) -- (v1);

\draw (v4) -- (v3);

%% Caption
\node[font=\small, thin] at (-2.9, -1.7)
    {\textit{Non-separable}};

%%============================================================
%% RIGHT: Separable (two triangles sharing the vertex x^*)
%%============================================================

%% Subgraph shading (draw fills before nodes so they stay under edges)
\fill[blue!10] (1.8,  1.0) -- (3.2, 0) -- (1.8, -1.0) -- cycle;
\fill[red!10]  (4.6,  1.0) -- (3.2, 0) -- (4.6, -1.0) -- cycle;

%% Vertices
\node[vtx]    (a)  at (1.8,  1.0) {};
\node[vtx]    (b)  at (1.8, -1.0) {};
\node[vtx]    (c)  at (4.6,  1.0) {};
\node[vtx]    (d)  at (4.6, -1.0) {};
\node[sepvtx] (xs) at (3.2,  0.0) {};   % separating vertex

%% Edges: left triangle + right triangle
\draw (a) -- (b) -- (xs) -- (a);
\draw (c) -- (d) -- (xs) -- (c);

%% Separating vertex label
\node[font=\scriptsize] at (3.25, 0.32) {$x^*$};

%% Subgraph labels
\node[font=\scriptsize, blue!60] at (2.3,  0.0) {$G_1$};
\node[font=\scriptsize, red!60]  at (4.1,  0.0) {$G_2$};

%% Caption
\node[font=\small, thin] at (3.2, -1.7) {\textit{Separable}};

%%============================================================
%% Vertical divider
%%============================================================
%\draw[gray!35, thin] (0, -1.4) -- (0, 1.7);

\end{tikzpicture}

\caption{\footnotesize{The left panel shows $K_4$ (complete graph on four vertices), which is 2-connected: removing any single vertex leaves the remaining graph connected, so no separating vertex exists. The right panel shows two triangles $G_1$ and $G_2$ joined at $x^*$ (shaded gray): removing $x^*$ disconnects $G_1$ from $G_2$, making $x^*$ a separating vertex and the graph separable.}}
\label{fig:separable}
\end{figure}
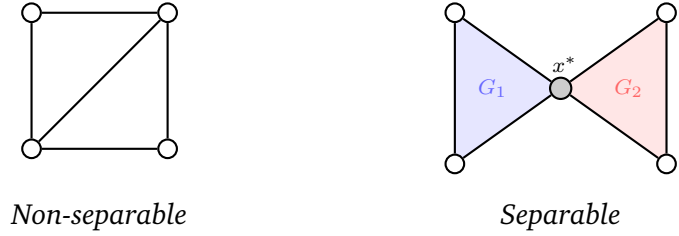

Remember $\Sigma=\{ \{1,3\},\{2,3\}\}$ induces $1\sim 3 \sim 2$, which is separable. On the other hand, $\Sigma'=\{ \{1,3\},\{2,3\}, \{1,2\} \}$ induces $1\sim 3 \sim 2 \sim 1$, which is non-separable. Example \ref{ex:non-identification} and the discussion after suggest that the identification is guaranteed as long as the induced graph is non-separable. The following example illustrates this point by constructing two distinct distributions over weights that generate identical behavior on all binary menus but differ when choices from larger menus are considered.

\begin{example}\label{ex:pappus}
    \label{Ex1} Let $\Sigma'=\{ \{1,3\},\{2,3\}, \{1,2\} \}$. The induced graph is non-separable.  Let $\pi$ and $\pi'$ be uniform distributions, with corresponding supports:
\begin{align*}
    \text{supp }\pi&= \{p^1 = \frac{1}{15}(2,3,10),\ p^2 = \frac{1}{13}(2,6,5),\ p^3 = \frac{1}{12}(4,3,5)\}, \\
\text{supp }\pi'&= \{q^1 = \frac{1}{18}(2,6,10),\ q^2 = \frac{1}{17}(4,3,10),\ q^3 = \frac{1}{15}(4,6,5)\}.
\end{align*}

Under $E=\{1,2\}$, $p^1(\cdot|E)=q^3(\cdot|E)$, $p^2(\cdot|E)=q^1(\cdot|E)$, and $p^3(\cdot|E)=q^2(\cdot|E)$. Since $\pi$ and $\pi'$ are uniform, they induce the same distribution under $E=\{1,2\}$. Similarly, under $E'=\{2,3\}$, $p^1(\cdot|E')=q^2(\cdot|E')$, $p^2(\cdot|E')=q^3(\cdot|E')$, and $p^3(\cdot|E')=q^1(\cdot|E')$. Finally, Hence, under $E''=\{1,3\}$, $p^1(\cdot|E'')=q^1(\cdot|E'')$, $p^2(\cdot|E'')=q^2(\cdot|E'')$, and $p^3(\cdot|E'')=q^3(\cdot|E'')$. Therefore, the sets of relative likelihoods agree for $\pi$ and $\pi'$ even though they are distinct distributions.\footnote{Note that $\pi$ and $\pi'$ are not unique; any convex combination $\hat{\pi} = \alpha \pi + (1 - \alpha) \pi'$ with $\alpha \in [0,1]$ will also satisfy the desired property.} Figure~\ref{fig:3-nonidentification} illustrates this example.  We refer to this example as a \emph{Pappus configuration}, as the hexagon described by the six beliefs relates to the classical Pappus' Theorem of projective geometry, see \emph{e.g.} \citet{coxeter1987}, Theorem 4.41 and the right hand side of Figure 4.4A.
\end{example}

\begin{figure}[h!]
    \centering
\begin{tikzpicture}[scale=1.6, >=stealth]

%% ---------------------------------------------------------------
%%  Coordinates.  Barycentric -> Cartesian (side = 4):
%%    X = 2(p_z - p_y),   Y = 2sqrt(3)*p_x
%%  Vertex x at top, y at bottom-left, z at bottom-right.
%% ---------------------------------------------------------------

\coordinate (Vx) at (0,      3.464);
\coordinate (Vy) at (-2,     0);
\coordinate (Vz) at (2,      0);

%% pi atoms
\coordinate (P1)  at ( 0.933,  0.462);  % (2,3,10)/15
\coordinate (P2)  at (-0.154,  0.533);  % (2,6,5)/13
\coordinate (P3)  at ( 0.333,  1.155);  % (4,3,5)/12

%% pi' atoms
\coordinate (P1p) at ( 0.444,  0.385);  % (2,6,10)/18
\coordinate (P2p) at ( 0.824,  0.815);  % (4,3,10)/17
\coordinate (P3p) at (-0.133,  0.924);  % (4,6,5)/15

%% Iso-conditional endpoints on edges
%% {x,y}-rays from Vz, on edge Vx-Vy  (p_z = 0)
\coordinate (Exy_A) at (-1.200, 1.386);  % p(x|{x,y}) = 2/5
\coordinate (Exy_B) at (-1.500, 0.866);  % p(x|{x,y}) = 1/4
\coordinate (Exy_C) at (-0.857, 1.979);  % p(x|{x,y}) = 4/7
%% {x,z}-rays from Vy, on edge Vx-Vz  (p_y = 0)
\coordinate (Exz_A) at ( 1.667, 0.577);  % p(x|{x,z}) = 1/6
\coordinate (Exz_B) at ( 1.429, 0.990);  % p(x|{x,z}) = 2/7
\coordinate (Exz_C) at ( 1.111, 1.540);  % p(x|{x,z}) = 4/9
%% {y,z}-rays from Vx, on edge Vy-Vz  (p_x = 0)
\coordinate (Eyz_A) at ( 1.077, 0);      % p(y|{y,z}) = 3/13
\coordinate (Eyz_B) at (-0.182, 0);      % p(y|{y,z}) = 6/11
\coordinate (Eyz_C) at ( 0.500, 0);      % p(y|{y,z}) = 3/8

%% ---------------------------------------------------------------
%%  Iso-conditional lines (background)
%% ---------------------------------------------------------------
\draw[orange, ultra thin, dashed]                  (Vz) -- (Exy_A);
\draw[orange, ultra thin, dashed]                  (Vz) -- (Exy_B);
\draw[orange, ultra thin, dashed]                  (Vz) -- (Exy_C);
\draw[cyan!55!black!55, very thin, dashed]           (Vy) -- (Exz_A);
\draw[cyan!55!black!55, very thin, dashed]           (Vy) -- (Exz_B);
\draw[cyan!55!black!55, very thin, dashed]           (Vy) -- (Exz_C);
\draw[red, ultra thin, dashed] (Vx) -- (Eyz_A);
\draw[red, ultra thin, dashed] (Vx) -- (Eyz_B);
\draw[red, ultra thin, dashed] (Vx) -- (Eyz_C);

%% ---------------------------------------------------------------
%% Event-edge labels
%% ---------------------------------------------------------------
\node[font=\scriptsize, cyan!80, rotate=-60] at (1.55, 2.15)
    {$E=\{1,3\}$};
\node[font=\scriptsize, red!80] at (0, -0.35)
    {$E'=\{2,3\}$};
    \node[font=\scriptsize, orange!80, rotate=60 ] at (-1.55, 2.15)
    {$E''=\{1,2\}$};

%% ---------------------------------------------------------------
%% Conditional-value labels at ray endpoints on edges
%% ---------------------------------------------------------------
% \filldraw[blue!55] (EA1) circle (1.5pt);
% \node[right=2pt, font=\tiny, blue!65] at (EA1)
%     {$\bigl(\tfrac{3}{4},\tfrac{1}{4}\bigr)$};

% \filldraw[blue!55] (EA2) circle (1.5pt);
% \node[right=2pt, font=\tiny, blue!65] at (EA2)
%     {$\bigl(\tfrac{1}{3},\tfrac{2}{3}\bigr)$};

% \filldraw[red!55] (EB1) circle (1.5pt);
% \node[below=2pt, font=\tiny, red!65] at (EB1)
%     {$\bigl(\tfrac{2}{3},\tfrac{1}{3}\bigr)$};

% \filldraw[red!55] (EB2) circle (1.5pt);
% \node[below=2pt, font=\tiny, red!65] at (EB2)
%     {$\bigl(\tfrac{1}{3},\tfrac{2}{3}\bigr)$};

%% ---------------------------------------------------------------
%% Point labels
%% ---------------------------------------------------------------

\node[right=2pt,  font=\scriptsize, blue!80] at (P1)   {$p^1$};
\node[left=2pt,  font=\scriptsize, blue!80] at (P2)   {$p^2$};
\node[above=2pt,  font=\scriptsize, blue!80] at (P3)   {$p^3$};
\node[below=2pt,  font=\scriptsize, green!50!black] at (P1p)   {$q^1$};
\node[right=2pt,  font=\scriptsize, green!50!black] at (P2p) {$q^2$};
\node[left=2pt,   font=\scriptsize, green!50!black] at (P3p) {$q^3$};

%% ---------------------------------------------------------------
%%  Pappus hexagon (sides alternate through Vx, Vy, Vz)
%%  p2 -- p3latex%% ---------------------------------------------------------------
%%  Pappus hexagon
%% ---------------------------------------------------------------
\draw[gray,  very thick, opacity=0.8]
    (P2) -- (P3p) -- (P3) -- (P2p) -- (P1) -- (P1p) -- cycle;

%% ---------------------------------------------------------------
%%  Simplex boundary
%% ---------------------------------------------------------------
\draw[very thick] (Vx) -- (Vy) -- (Vz) -- cycle;

%% ---------------------------------------------------------------
%%  Vertex labels
%% ---------------------------------------------------------------
\node[above, font=\small\bfseries] at (Vx) {$1$};
\node[left,  font=\small\bfseries] at (Vy) {$2$};
\node[right, font=\small\bfseries] at (Vz) {$3$};

%% ---------------------------------------------------------------
%%  pi atoms: orange upward triangles
%% ---------------------------------------------------------------

%% ---------------------------------------------------------------
%% Atoms of π  (blue filled discs; larger = more weight)
%% ---------------------------------------------------------------
\filldraw[fill=blue!25, draw=blue!80, thick] (P1) circle (1pt); % w=1/3
\filldraw[fill=blue!25, draw=blue!80, thick] (P2) circle (1pt); % w=1/3
\filldraw[fill=blue!25, draw=blue!80, thick] (P3) circle (1pt); % w=1/3

\draw[fill=green!50!black, green!50!black, very thick] (P1p) circle (1pt);            % q^1, w=1/3
\filldraw[fill=green!50!black, draw=green!50!black, thick]
    (P2p) circle (1.0pt);                                         % q^2, w=1/3
\filldraw[fill=green!50!black, draw=green!50!black, thick]
    (P3p) circle (1.0pt);

\end{tikzpicture}
\caption{\footnotesize{Pappus configuration in $\Delta(\{x,y,z\})$.  Blue discs are the atoms of $\pi$ (area proportional to weight), and green discs are the atoms of $\pi'$.  The three families of iso-conditional
rays encode the binary-menu observables.  In each family,
every ray contains exactly one atom of $\pi$ and one of $\pi'$, so the
two distributions are indistinguishable on all binary menus.  The gray
hexagon $p^2,q^3,p^3,q^2,p^1,q^1$ is the Pappus hexagon.
%its six sides alternate through vertices $1$, $2$, $3$, and each pair of opposite sides concurs at the same vertex. 
}}
\label{fig:3-nonidentification}
\end{figure}
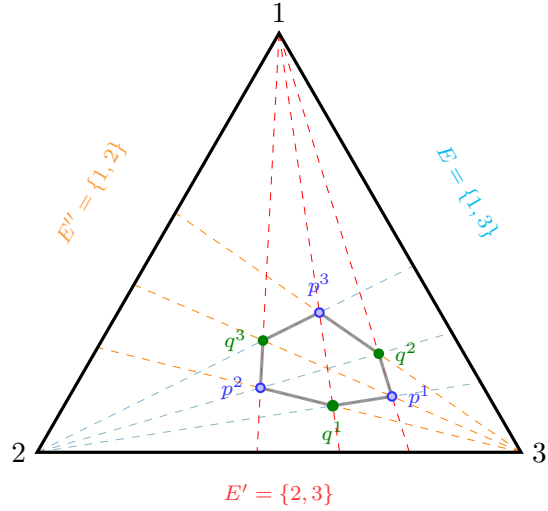

The example illustrates that non-separability of the induced graph does not guarantee identification for \emph{every} $\pi\in\Delta^n(\Delta_{++}(X))$. The Pappus configuration provides a distribution that evades identification even when $\Sigma'$ contains all binary events and the induced graph is non-separable ($K_3$, which is 2-connected).  In the next section, our main result illustrates that identified distributions form a \emph{generic} subset, not that every distribution is identified.  The Pappus example is therefore a knife-edge exception.  This motivates the notion of genericity.  We say a subset of a topological space is \emph{generic} if it contains an open and dense subset; a property that fails only on a non-generic set is thus exceptional in a strong topological sense. Intuitively, a generic set is a large subset of the grand set from the topological point of view: it contains an open region around each of its own points, and it comes arbitrarily close to every point of the grand set.

  The answer turns entirely on a single structural property of the
family of events: whether the induced $(X,\sim_{\Sigma})$ forms a separable graph or not.  Our main result shows that these two cases yield opposite conclusions.  When $(X,\sim)$ is non-separable, identification is generic: the set of $n$-agent distributions
that are identified by $\Sigma$ is large in a topological sense, so
non-identification is the exceptional case.  When $(X,\sim)$ is separable,
the conclusion reverses and non-identification becomes generic, consistent
with the intuition from Example~2 that richer separation of states allows
more room for distinct distributions to produce identical observables.

\begin{theorem}\label{thm:separable}Let $n\geq 2$.
\begin{enumerate}
\item If $(X,\sim)$ is non-separable, then $\{\pi\in\Delta^n(\Delta_{++}(X)):\pi\text{ is identified by }\Sigma\}$ is generic.
\item If $(X,\sim)$ is separable, then $\{\pi\in\Delta^n(\Delta_{++}(X)):\pi\text{ is not identified by }\Sigma\}$ is generic.
\end{enumerate}
\end{theorem}

Two remarks about Theorem~\ref{thm:separable} are in order:

\begin{itemize}
\item Theorem~\ref{thm:separable} illustrates  that the number of sets on which one has information about marginals is less important than the underlying graph-theoretic structure in determining identification.  Thus, even when all sets have the same cardinality (binary sets for example), naively counting them as ``data'' may lead to misleading results.  For example, the separable graph in Figure~\ref{fig:separable} arises from six binary sets on five states, and is generically non-identified.  On the other hand, a simple cycle through the five states arises from five binary sets, and is generically identified.
\item A refinement of the theorem is available:  the generic number of rationalizations---namely the cardinality of the set of $\pi'$ which has the same induced Bayesian updates as a given $\pi$---for a given environment with $n\geq 1$ is $(n!)^{m-1}$, where $m$ refers to the number of ``nonseparable components,'' or maximal nonseparable subgraphs (here, any connected pair of states is by definition a nonseparable graph).  Thus, $(n!)^{m-1}$ can be understood as a ``degree of non-identification.''  See \emph{e.g.} \citet{harary1959} for a method of computing $m$.  This allows us to infer that graphs forming tree structures are generically the worst from an identification perspective (trees have $m=|X|-1$).  This is to be expected as trees are the minimal graph-theoretic structures allowing identification in the single-agent case.
\end{itemize}

The proof rests on two lemmas.  The first is a folk result characterizing probabilistic consistency via cycle products, utilizing an axiom found already in \citet{luce1959individual}, who established one direction of the result.  In log form, it is a result describing the existence of potential functions.  We utilize an argument similar to the one found in \citet{afriat1963} or \citet{rockafellar1966}.  The second provides a topological tool for transferring density and openness from $\Delta_{++}(X)^n$ to $\Delta^n(\Delta_{++}(X))$.

%%==========================================================================
%% LEMMA 1: PROBABILISTIC CONSISTENCY
%%==========================================================================

A \emph{setwise cycle} for the environment $(X,\Sigma)$ is a list
\[
  \mathcal{C} = \bigl((x_1,x_2),E_1;\; (x_2,x_3),E_2;\; \ldots;\; (x_m,x_1),E_m\bigr),
\]
where the vertices $x_1,\ldots,x_m\in X$ are distinct, each $E_i\in\Sigma$, and $\{x_i,x_{i+1}\}\subseteq E_i$ for each $i$ (indices mod $m$).

For a collection $\{p_E\}_{E\in\Sigma}$ with each $p_E\in\Delta_{++}(E)$, and a setwise cycle $\mathcal{C}$ as above, define the \emph{product along $\mathcal{C}$} as:
\[
  \Pi(\mathcal{C},\{p_E\}) \;=\; \prod_{i=1}^{m} \frac{p_{E_i}(x_{i+1})}{p_{E_i}(x_i)},
\]
where $x_{m+1}:=x_1$.

We say $\{p_E\}_{E\in\Sigma}$ is \emph{probabilistically consistent} if there exists $p\in\Delta_{++}(X)$ such that $p_E=p(\cdot|E)$ for all $E\in\Sigma$.

\begin{lemma}\label{lem:consistent}
The collection $\{p_E\}_{E\in\Sigma}$ is probabilistically consistent if and only if $\Pi(\mathcal{C},\{p_E\})=1$ for every setwise cycle $\mathcal{C}$.
\end{lemma}

\begin{proof}
Probabilistic consistency is equivalent to the existence of $u\in\mathbb{R}^X_{++}$ such that for all $x,y\in X$ and all $E\in\Sigma$ with $\{x,y\}\subseteq E$,
\begin{equation}\label{eq:ratio}
  \frac{u(x)}{u(y)} = \frac{p_E(x)}{p_E(y)}.
\end{equation}
%(Given such $u$, set $p(x)=u(x)/\sum_{z\in X}u(z)$; then $p\in\Delta_{++}(X)$ and $p(x|E) = u(x)/\sum_{y\in E}u(y) = p_E(x)$, since $p_E$ is itself a probability on $E$ and all ratios $p_E(x)/p_E(y)$ are determined by $u$.)

\medskip
\noindent\textit{Necessity (consistency $\Rightarrow$ unit products).}  This follows immediately owing to the fact that $\frac{p_{E}(x)}{p_E(y)}=\frac{p(x)}{p(y)}$ for a probabilistically consistent collection.%If $u$ satisfying \eqref{eq:ratio} exists, then along any setwise cycle $((x_1,x_2),E_1;\ldots;(x_m,x_1),E_m)$,
%\[
%  \prod_{i=1}^m \frac{p_{E_i}(x_{i+1})}{p_{E_i}(x_i)}
%  = \prod_{i=1}^m \frac{u(x_{i+1})}{u(x_i)}
%  = 1,
%\]
%where the last equality holds by telescoping (the product is cyclic).

\medskip
\noindent\textit{Sufficiency (unit products $\Rightarrow$ consistency).}  We work with connected graphs; for disconnected graphs apply the same argument to each connected component and normalize each component independently.

\emph{Extension to generalized cycles.}  Call a list $((x_1,x_2),E_1;\ldots;(x_m,x_1),E_m)$ a \emph{generalized setwise cycle} if it satisfies all conditions of a setwise cycle except that the $x_i$'s need not be distinct.  Products along generalized cycles are defined in the same way.  We claim that products along setwise cycles equal $1$ if and only if all products along generalized setwise cycles equal $1$:  Take any generalized setwise cycle where elements repeat and take the product; inductively replace any minimal segment where one element repeats by $1$.

\emph{Construction of $u$.}  Fix any $x^*\in X$ and set $u(x^*)=1$.  For any $x\in X\setminus\{x^*\}$, fix a path $(x^*=x_1,E_1; x_2,E_2;\ldots; x_{k-1},E_{k-1}; x_k=x)$, where $\{x_i,x_{i+1}\}\subseteq E_i\in\Sigma$ for each $i$, and define
\[
  u(x) = \prod_{i=1}^{k-1} \frac{p_{E_i}(x_{i+1})}{p_{E_i}(x_i)} > 0.
\]
Such a path exists by connectedness of $(X,\sim)$.  To see $u(x)$ is independent of the choice of path: given two paths $\gamma_1$ and $\gamma_2$ from $x^*$ to $x$, concatenating $\gamma_1$ with the reversal of $\gamma_2$ yields a generalized setwise cycle.  Since all products along generalized setwise cycles equal $1$, both paths assign the same value to $u(x)$.

\emph{Verification.}  Fix any $E\in\Sigma$ and $x,y\in E$.  Let $\gamma = (x^*=x_1,E_1;\ldots;x_{k-1},E_{k-1};x_k=x)$ be the path used to define $u(x)$.  Then
\[
  (x^*=x_1,E_1;\ldots; x_{k-1},E_{k-1}; x_k=x, E; y)
\]
is a valid path from $x^*$ to $y$, giving $u(y) = u(x)\frac{p_E(y)}{p_E(x)}$, i.e., $\frac{u(x)}{u(y)}=\frac{p_E(x)}{p_E(y)}$, which is \eqref{eq:ratio}. \hfill$\square $\end{proof}

%%==========================================================================
%% LEMMA 2: HOMEOMORPHISM
%%==========================================================================

Next, let us define $\psi:\Delta(X)^n\rightarrow\Delta^n(\Delta(X))$ by $\pi(p^1,\ldots,p^n)=\frac{1}{n}\sum_{i=1}^n \delta_{p_i}$.  Define an equivalence relation $\equiv$ via $(p_1,\ldots,p_n)\equiv(q_1,\ldots,q_n)$ if there exists a permutation $\sigma$ of $\{1,\ldots,n\}$ with $p_i=q_{\sigma(i)}$ for all $i$.  As is usual, for $(p_1,\ldots,p_n)\in\Delta(X)^n$, define $[(p_1,\ldots,p_n)]\in\Delta(X)^n/\equiv$ to be its equivalence class, and define $Q:\Delta(X)^n\rightarrow\Delta(X)^n/\equiv$ to be the quotient map.  Since $\psi$ is constant on equivalence classes, we can write $\psi = \psi^* \circ Q$ for a well-defined map $\psi^*([p^1,\ldots,p^n])=\sum_i \delta_{p^i}$.

\begin{lemma}\label{lem:homeo}
$\Delta^n(\Delta(X))$ with the relative weak topology is homeomorphic to $\Delta(X)^n/{\equiv}$ under $\psi^*$.
\end{lemma}

\begin{proof}
%Define $\psi:\Delta(X)^n\rightarrow\Delta^n(\Delta(X))$ by $\psi(p_1,\ldots,p_n)=\frac{1}{n}\sum_{i=1}^n \delta_{p_i}$.  

The function $\psi$ is continuous: if $(p_1^k,\ldots,p_n^k)\rightarrow(q_1,\ldots,q_n)$ in $\Delta(X)^n$, then for any bounded continuous $f:\Delta(X)\rightarrow\mathbb{R}$,
\[
  \int f\,d\psi(p_1^k,\ldots,p_n^k) = \frac{1}{n}\sum_{i=1}^n f(p_i^k) \;\longrightarrow\; \frac{1}{n}\sum_{i=1}^n f(q_i) = \int f\,d\psi(q_1,\ldots,q_n).
\]
Then the map $\psi^*:\Delta(X)^n/{\equiv}\;\rightarrow\;\Delta^n(\Delta(X))$ via $\psi^*([(p_1,\ldots,p_n)])=\frac{1}{n}\sum_i\delta_{p_i}$, is continuous by definition.

%\emph{Surjectivity.}  Given any $\pi\in\Delta^n(\Delta(X))$, write its support as $\{r_1,\ldots,r_s\}$ with weights $\pi(\{r_j\})=m_j/n$ (each $m_j\in\{1,\ldots,n\}$, $\sum_j m_j=n$).  List each $r_j$ exactly $m_j$ times to form a tuple $(p_1,\ldots,p_n)$; then $\psi(p_1,\ldots,p_n)=\pi$.

%\emph{Injectivity.}  If $\psi^*([(p_1,\ldots,p_n)])=\psi^*([(q_1,\ldots,q_n)])$, then $\frac{1}{n}\sum_i\delta_{p_i}=\frac{1}{n}\sum_i\delta_{q_i}$ as measures, so the multisets $\{p_1,\ldots,p_n\}$ and $\{q_1,\ldots,q_n\}$ coincide, i.e., $(p_1,\ldots,p_n)\equiv(q_1,\ldots,q_n)$.

\emph{Homeomorphism.}  $\Delta(X)^n$ is compact, so $\Delta(X)^n/{\equiv}$ is compact (the quotient map is continuous and surjective).  The weak topology on $\Delta(\Delta(X))$ is metrizable (hence Hausdorff, see \citet{aliprantis2006} Theorem 15.11), so its subspace $\Delta^n(\Delta(X))$ is Hausdorff.  A continuous bijection from a compact space to a Hausdorff space is a homeomorphism:  \citet{willard1970}, Theorem 17.14.   \hfill$\square $ \end{proof}

\begin{corollary}\label{cor:quotient}
The restriction $\psi^*\big|_{\Delta_{++}(X)^n/\equiv}:\Delta_{++}(X)^n/{\equiv}\;\rightarrow\;\Delta^n(\Delta_{++}(X))$ is a homeomorphism between the respective relative topologies.
\end{corollary}

\begin{proof}See, \emph{e.g.} Theorem 7.5 of \citet{willard1970}.\end{proof}

%%==========================================================================
%% LEMMA 3: QUOTIENT MAP IS OPEN
%f
%==========================================================================

\begin{lemma}\label{lem:openmap}
The quotient map $Q:\Delta(X)^n\rightarrow\Delta(X)^n/{\equiv}$ is open; that is, for any open $U$, $Q(U)$ is open.
\end{lemma}

\begin{proof}
Let $U\subseteq\Delta(X)^n$ be open.  We show $Q(U)$ is open in the quotient topology, i.e., $Q^{-1}(Q(U))$ is open in $\Delta(X)^n$.  An element $(p_1,\ldots,p_n)$ belongs to $Q^{-1}(Q(U))$ if and only if its equivalence class intersects $U$, i.e., if and only if $(p_{\sigma(1)},\ldots,p_{\sigma(n)})\in U$ for some permutation $\sigma$.  Therefore
\[
  Q^{-1}(Q(U)) = \bigcup_{\sigma\in S_n} U\circ\sigma,
\]
where $U\circ\sigma = \{(p_{\sigma(1)},\ldots,p_{\sigma(n)}):(p_1,\ldots,p_n)\in U\}$ and $S_n$ denotes the set of permutations from $\{1,\ldots,n\}$ to $\{1,\ldots,n\}$.  Each $U\circ\sigma$ is open (pre-composing a tuple by a fixed permutation is a homeomorphism of $\Delta(X)^n$ to itself), so $Q^{-1}(Q(U))$ is open as a finite union of open sets.  \hfill$\square $ \end{proof}

%%==========================================================================
%% PROPOSITION: OPENNESS AND DENSITY TRANSFER
%%==========================================================================

\begin{proposition}\label{prop:topology}
If $A\subseteq\Delta_{++}(X)^n$ is open and dense (in the subspace topology), then $\psi(A)\subseteq\Delta^n(\Delta_{++}(X))$ is open and dense.
\end{proposition}

\begin{proof}
%Write $\pi^* = \psi^*\circ Q$, so $\pi^*|_{\Delta_{++}(X)^n} = \psi|_{\Delta_{++}(X)^n}$.{\color{red}CHRIS:  NO, WE HAVE BEEN USING $\pi$ FOR PROBABILITY MEASURES, I DON'T KNOW WHAT THIS MEANS}

\emph{Openness.}  A homeomorphism is by definition an open map, and $Q$ is an open map, so $\psi=\psi^*\circ Q$ is an open map.

%Since $Q$ is open (Lemma~\ref{lem:openmap}), $Q(A)$ is open in $\Delta(X)^n/{\equiv}$.  Since $\psi^*$ restricted to $\Delta_{++}(X)^n/{\equiv}$ is a homeomorphism (Corollary~\ref{cor:quotient}), $\psi^*(Q(A)\cap(\Delta_{++}(X)^n/{\equiv}))$ is open in $\Delta^n(\Delta_{++}(X))$, and this equals $\psi(A)$.

\emph{Density.}  This follows from a standard characterization of continuity and the fact that $\psi$ is surjective.  Let $A\subseteq \Delta_{++}(X)^n$ be dense, so that $\overline{A}=\Delta_{++}(X)^n$.  Then $\Delta^n(\Delta_{++}(X))=\psi(\Delta_{++}(X)^n)=\psi(\overline{A})\subseteq \overline{\psi(A)}\subseteq\Delta^n(\Delta_{++}(X))$, where the first equality owes to surjectivity and the first set inclusion is by \citet{willard1970} Theorem 7.2d.  Thus $\overline{\psi(A)}=\Delta^n(\Delta_{++}(X))$ and $\psi(A)$ is dense.

\end{proof}
%For a continuous surjection $f:B\twoheadrightarrow C$ and a dense set $A\subseteq B$: given any nonempty open $V\subseteq C$, the preimage $f^{-1}(V)$ is open and nonempty (by surjectivity), so $A\cap f^{-1}(V)\neq\varnothing$, and thus $f(A)\cap V\neq\varnothing$.  Apply this with $f = \psi|_{\Delta_{++}(X)^n}$ (which is continuous and surjective onto $\Delta^n(\Delta_{++}(X))$) and $A$ dense in $\Delta_{++}(X)^n$.  \end{proof}

%%==========================================================================
%% PROOF OF MAIN THEOREM
%%==========================================================================

\begin{proof}[Proof of Theorem~\ref{thm:separable}.]

By Proposition~\ref{prop:topology}, it suffices in each case to exhibit an open dense subset of $\Delta_{++}(X)^n$ on which the relevant identification property holds, and then apply $\psi$ to transfer the result to $\Delta^n(\Delta_{++}(X))$.

It is straightforward to show that for $(p^1,\ldots,p^n)\in\Delta_{++}(X)^n$, $\psi(p^1,\ldots,p^n)$ is identified iff whenever $(q^1,\ldots,q^n)\in\Delta_{++}(X)^n$ and for all $E\in\Sigma$, $(q^1(\cdot|E),\ldots,q^n(\cdot|E))\equiv (p^1(\cdot|E),\ldots,p^n(\cdot|E))$, then $(q^1,\ldots,q^n)\equiv (p^1,\ldots,p^n)$.  

\medskip
\noindent\textbf{Case 2: $(X,\sim)$ is separable.}

Let $x^*$ be a separating vertex and let $X_1,X_2$ be the two connected subgraphs of the separation, with $X_1\cap X_2=\{x^*\}$.  Both $X_1$ and $X_2$ have cardinality at least $2$ (since $x^*$ is a non-isolated separating vertex).

\emph{Step 1: Every event is contained in $X_1$ or in $X_2$.}  Suppose $E\in\Sigma$ satisfies $E\cap X_1\neq\varnothing$ and $E\cap X_2\neq\varnothing$.  We claim either $E\cap(X_1\setminus\{x^*\})=\varnothing$ or $E\cap(X_2\setminus\{x^*\})=\varnothing$, i.e., $E\subseteq X_1$ or $E\subseteq X_2$.  Suppose for contradiction that there exist $x_1\in E\cap(X_1\setminus\{x^*\})$ and $x_2\in E\cap(X_2\setminus\{x^*\})$.  Then $\{x_1,x_2\}\subseteq E$, so $x_1\sim x_2$.  But $x_1\in X_1\setminus\{x^*\}$ and $X_1,X_2$ share only $x^*$, so any neighbor of $x_1$ is in $X_1$ or equals $x^*$---contradicting $x_2\in X_2\setminus\{x^*\}$.

\emph{Step 2: The mixing operation $p_{X_1}q$.}  For $p,q\in\Delta_{++}(X)$, define $p_{X_1}q\in\Delta_{++}(X)$ by
\[
p_{X_1}q(x) = \begin{cases}
  \dfrac{p(x) \, q(x^*)}{q(x^*)\,p(X_1)+p(x^*)\,q(X_2)-p(x^*)\,q(x^*)} & x\in X_1\setminus\{x^*\}, \\[6pt]
  \dfrac{p(x^*)\,q(x)}{q(x^*)\,p(X_1)+p(x^*)\,q(X_2)-p(x^*)\,q(x^*)} & x\in X_2\setminus\{x^*\}, \\[6pt]
  \dfrac{p(x^*)\,q(x^*)}{q(x^*)\,p(X_1)+p(x^*)\,q(X_2)-p(x^*)\,q(x^*)} & x=x^*,
\end{cases}
\]
where $p(X_j) = \sum_{x\in X_j} p(x)$. % (One checks the weights sum to $1$.)

The key property of this operation is:
\begin{itemize}
  \item For all $E\in\Sigma$ with $E\subseteq X_1$: $\;p_{X_1}q(\cdot|E) = p(\cdot|E)$.
  \item For all $E\in\Sigma$ with $E\subseteq X_2$: $\;p_{X_1}q(\cdot|E) = q(\cdot|E)$.
\end{itemize}
To verify the first: for $x\in E\subseteq X_1$,
\[
  p_{X_1}q(x|E) = \frac{p_{X_1}q(x)}{\sum_{y\in E}p_{X_1}q(y)}
  = \frac{q(x^*)p(x)/D}{\sum_{y\in E}q(x^*)p(y)/D}
  = \frac{p(x)}{p(E)} = p(x|E),
\]
where $D = q(x^*)p(X_1)+p(x^*)q(X_2)-p(x^*)q(x^*)$.  The second assertion is analogous with $p$ and $q$ exchanged on $X_2$.

\emph{Step 3: Generic non-identification.}  Let
\[
  \Delta_{++}^{X_1,X_2}(X)^n = \bigl\{(p^1,\ldots,p^n)\in\Delta_{++}(X)^n : \forall\,i\neq j,\; p^i(\cdot|X_1)\neq p^j(\cdot|X_1)\;\text{and}\;p^i(\cdot|X_2)\neq p^j(\cdot|X_2)\bigr\}.
\]
This set is open (each of the finite list of conditions $p^i(\cdot|X_j)\neq p^k(\cdot|X_j)$ defines an open set) and dense (each condition fails only on a hyperplane of lower dimension) in $\Delta_{++}(X)^n$.

For any $(p^1,\ldots,p^n)\in\Delta_{++}^{X_1,X_2}(X)^n$, consider the cyclically shifted tuple $(p^1_{X_1}p^2, p^2_{X_1}p^3, \ldots, p^{n-1}_{X_1}p^n, p^n_{X_1}p^1)$.  By the separation of $\Sigma$ shown in Step~1 and the conditional preservation shown in Step~2:
\begin{itemize}
  \item $\psi(p^1_{X_1}p^2,\ldots,p^n_{X_1}p^1)$ and $\psi(p^1,\ldots,p^n)$ induce the same distribution over $\{p(\cdot|E):p\in\text{supp}\}$ for every $E\in\Sigma$.
  \item However, $\psi(p^1_{X_1}p^2,\ldots,p^n_{X_1}p^1)\neq\psi(p^1,\ldots,p^n)$: since each $p^i_{X_1}p^{i+1}$ inherits its $X_1$-marginal from $p^i$ and its $X_2$-marginal from $p^{i+1}$, and by the separation condition these are all distinct, the tuple $(p^i_{X_1}p^{i+1})_i$ is not a permutation of $(p^i)_i$.
\end{itemize}
Hence $\psi(p^1,\ldots,p^n)$ is not identified.  By Proposition~\ref{prop:topology}, $\psi(\Delta_{++}^{X_1,X_2}(X)^n)$ is open and dense in $\Delta^n(\Delta_{++}(X))$, and every element is not identified.

\medskip
\noindent\textbf{Case 1: $(X,\sim)$ is non-separable.}

Since $(X,\sim)$ is non-separable (2-connected), by a classical result of \citet{whitney1931} (see also \citealt{bondy2008}, Theorem~9.2), every pair of edges of $(X,\sim)$ lie on a cycle, where a cycle is a list of distinct vertices $(x^1,\ldots,x^k)$ with $x^i \sim x^{i+1}$ (here $k+1=1$, addition is mod $k$).% $\{x,y\}$ (i.e., any pair with $x\sim y$), there exists a cycle $(x^1,\ldots,x^k)$ of distinct vertices with $x^i\sim x^{i+1}$ (mod $k$) that contains the edge $\{x,y\}$.

\emph{Step 1: Non-identification implies a labeled cycle with unit product.}  Consider the subset
\[
  \Delta_{++}^u(X)^n = \bigl\{(p^1,\ldots,p^n)\in\Delta_{++}(X)^n : p^i\neq p^j\;\forall\,i\neq j\bigr\},
\]
which is open and dense in $\Delta_{++}(X)^n$.

Fix $(p^1,\ldots,p^n)\in\Delta_{++}^u(X)^n$ and suppose $\pi=\psi(p^1,\ldots,p^n)$ is not identified.  Then by definition, there is $(q^1,\ldots,q^n)\in\Delta_{++}(X)^n$ for which $\psi(p^1,\ldots,p^n)\neq \psi(q^1,\ldots,q^n)$, but for all $E\in\Sigma$, $\psi(q^1(\cdot|E),\ldots,q^n(\cdot|E))=\psi(p^1(\cdot|E),\ldots,p^n(\cdot|E))$.  In particular, there is some $t\in\{1,\ldots,n\}$ for which $p^t\notin \{q^1,\ldots,q^n\}$, as otherwise $\psi(p^1,\ldots,p^n)=\psi(q^1,\ldots,q^n)$ (owing to the fact that $(p^1,\ldots,p^n)\in\Delta_{++}^u(X)^n$).  Therefore, there is some $t'\neq t$ and $q\in\{q^1,\ldots,q^n\}$ and a pair of distinct events $E^t\in\Sigma$ and $E^{t'}\in\Sigma$ for which $q(\cdot|E^t)=p^t(\cdot|E^t)$ and $q(\cdot|E^{t'})=p^{t'}(\cdot|E^{t'})$.

%Then there exists $\pi'\neq\pi$ that agrees with $\pi$ on every event-conditional.  Writing $\pi'=\psi(q^1,\ldots,q^n)$ (for some tuple, possibly in a different order), and using the fact that all $p^i$ are distinct, there must exist some index $l$ and some $q\in\Delta_{++}(X)$ with $q\neq p^l$ such that for every $E\in\Sigma$, $q(\cdot|E)$ agrees with some $p^j(\cdot|E)$---but the assignment $E\mapsto j(E)$ is not constant (otherwise $q$ would equal some fixed $p^j$ by Proposition~\ref{prop:connected}).

Conclude that there exists a nonempty subset $L\subseteq\{1,\ldots,n\}$ with $|L|\geq 2$ and a labeled partition $\{\Sigma^l\}_{l\in L}$ of $\Sigma$ such that $q(\cdot|E)=p^l(\cdot|E)$ for all $E\in\Sigma^l$, where the labeling is non-constant.  In particular, for every edge $\{x,y\}$ appearing in some $E\in\Sigma$, we obtain a label $l(\{x,y\})\in L$ with $q(\cdot|\{x,y\})=p^{l(\{x,y\})}(\cdot|\{x,y\})$, and again where the map carrying edges to $L$ is nonconstant.

Let $\{x,y\}$ and $\{z,w\}$ be any distinct pair of edges corresponding to different labels.  Because $(X,\sim)$ is non-separable, by \citet{bondy2008} Theorem 5.2 there is a cycle $\mathcal{C}=(x^1,\ldots,x^k)$ containing each of these edges.  Traversing $\mathcal{C}$ and assigning to each edge $(x^i,x^{i+1})$ the label $\phi(x^i,x^{i+1})=l(\{x^i,x^{i+1}\})\in\{1,\ldots,n\}$ defines a \emph{labeled cycle} $(\mathcal{C},\phi)$.  Since the overall labeling is non-constant, we can choose $\mathcal{C}$ so that $\phi$ is non-constant on its edges.

By Lemma~\ref{lem:consistent} applied to $q$: since $q\in\Delta_{++}(X)$, the collection $\{q(\cdot|E)\}_{E\in\Sigma}$ is trivially probabilistically consistent (it comes from $q$).  For each edge $(x^i,x^{i+1})$ of $\mathcal{C}$, let $l_i = \phi(x^i,x^{i+1})$; then $q(\cdot|\{x^i,x^{i+1}\}) = p^{l_i}(\cdot|\{x^i,x^{i+1}\})$, which gives $q(x^i)/q(x^{i+1})=p^{l_i}(x^i)/p^{l_i}(x^{i+1})$.  Applying the cycle product formula from Lemma~\ref{lem:consistent} to $q$:
\[
  1 = \prod_{i=1}^k \frac{q(x^{i+1})}{q(x^i)} = \prod_{i=1}^k \frac{p^{l_i}(x^{i+1})}{p^{l_i}(x^i)}.
\]
That is, the labeled cycle product
\begin{equation}\label{eq:cycle-unit}
  \prod_{i=1}^k \frac{p^{\phi(x^i,x^{i+1})}(x^{i+1})}{p^{\phi(x^i,x^{i+1})}(x^i)} = 1.
\end{equation}

\emph{Step 2: A generic set on which all labeled cycle products differ from $1$.}  There are only finitely many labeled cycles $(\mathcal{C},\phi)$ (since $X$ and $n$ are finite).  For each such labeled cycle with $\phi$ non-constant, the set
\[
  A_{(\mathcal{C},\phi)} = \biggl\{(p^1,\ldots,p^n)\in\Delta_{++}^u(X)^n : \prod_{i=1}^k \frac{p^{\phi(x^i,x^{i+1})}(x^{i+1})}{p^{\phi(x^i,x^{i+1})}(x^i)} \neq 1\biggr\}
\]
is open and dense in $\Delta_{++}(X)^n$.

Since there are finitely many such labeled cycles, the intersection
\[
  A^* = \bigcap_{(\mathcal{C},\phi):\,\phi\text{ non-constant}} A_{(\mathcal{C},\phi)}
\]
is also open and dense in $\Delta_{++}(X)^n$.

\emph{Step 3: Every element of $A^*$ gives an identified distribution.}  By Step~1, if $(p^1,\ldots,p^n)\in\Delta_{++}^u(X)^n$ and $\psi(p^1,\ldots,p^n)$ is not identified, then there exists a non-constant labeled cycle $(\mathcal{C},\phi)$ for which \eqref{eq:cycle-unit} holds.  Thus, for every $(p^1,\ldots,p^n)\in A^*$, $\psi(p^1,\ldots,p^n)$ is identified.

By Proposition~\ref{prop:topology}, $\psi(A^*)$ is open and dense in $\Delta^n(\Delta_{++}(X))$, and every element is identified.  \end{proof}

\section{(Non-)Identification on $\Delta^S(\Delta_{++}(X))$}

Having characterized identification on the class $\Delta^n(\Delta_{++}(X))$ of $n$-agent distributions, we now turn to the broader class $\Delta^S(\Delta_{++}(X))$ of all finitely-supported distributions, endowed with the weak topology.  Enlarging the domain in this way fundamentally alters the identification landscape.  In the $n$-agent case, non-separability of $\sim_{\Sigma}$ was enough to make identification a generic property. That is, the identified distributions formed an open and dense subset.  On $\Delta^S(\Delta_{++}(X))$, this clean dichotomy breaks down: even under non-separability, the identified and non-identified distributions coexist as dense sets, so neither is topologically negligible.  The reason is that the Pappus-type non-identified distributions constructed in
Theorem~\ref{thm:nonidentify} can be mixed with any target distribution at
an arbitrarily small weight, placing non-identified distributions
arbitrarily close to every point of the space.  Separability, on the other
hand, still delivers a strong conclusion: the set of non-identified
distributions is open and dense, so identification becomes the genuinely
exceptional case.

% Let us consider now $\Delta^S(\Delta_{++}(X))$, again endowed with the weak topology.  

% The following result shows that identification no longer remains a generic property when $\sim_{\Sigma}$; though it is still obtained on a dense set.  

\medskip
\noindent\textbf{Nonseparable case.}  Let $X=\{1,2,3\}$ and 
$\Sigma=\{\{1,2\},\{1,3\},\{2,3\}\}$, whose induced graph is $K_3$
(non-separable), with $X\notin\Sigma$.  Consider
$\pi^*=\delta_{(1/3,1/3,1/3)}\in\Delta^S(\Delta_{++}(X))$.  Since every
binary conditional of $(1/3,1/3,1/3)$ equals $\tfrac{1}{2}$, any
$\pi'\in\Delta^S(\Delta_{++}(X))$ producing the same observables must also
assign all binary conditionals the value $\tfrac{1}{2}$ with probability
one, which forces every atom of $\pi'$ to be $(1/3,1/3,1/3)$.  Hence
$\pi^*$ is identified.

To see that non-identified distributions are dense near $\pi^*$, let
$\pi_0$ and $\pi_0'$ be the two distributions from Example~\ref{ex:pappus}
(the Pappus configuration), which satisfy
$\pi_0\neq\pi_0'$ yet induce the same conditional distributions on every
$E\in\Sigma$.  For each $\varepsilon\in(0,1)$ define
\[
  \pi^\varepsilon
  \;=\;(1-\varepsilon)\,\pi^*+\varepsilon\,\pi_0,
  \qquad
  \tilde\pi^\varepsilon
  \;=\;(1-\varepsilon)\,\pi^*+\varepsilon\,\pi_0'.
\]
Since $\pi_0$ and $\pi_0'$ produce the same observables and
$(1-\varepsilon)\pi^*$ is common to both, $\pi^\varepsilon$ and
$\tilde\pi^\varepsilon$ induce the same conditional distributions on every
$E\in\Sigma$, yet $\pi^\varepsilon\neq\tilde\pi^\varepsilon$.  Thus
$\pi^\varepsilon$ is not identified for any $\varepsilon>0$, while
$\pi^\varepsilon\to\pi^*$ (identified) as $\varepsilon\to 0$.  This
confirms that both identified and non-identified distributions accumulate at
every point of $\Delta^S(\Delta_{++}(X))$.

\medskip
\noindent\textbf{Separable case.}  Let $\Sigma'=\{\{1,2\},\{2,3\}\}$,
whose induced graph is the path $1 \sim_\Sigma 2 \sim_\Sigma 3$ with separating vertex~$2$,
giving components $X_1=\{1,2\}$ and $X_2=\{2,3\}$.  Consider
\[
  \pi \;=\; \tfrac{1}{2}\,\delta_{p}+\tfrac{1}{2}\,\delta_{q},
  \qquad
  p=\tfrac{1}{4}(2,1,1),\quad q=\tfrac{1}{4}(1,1,2).
\]
The binary conditionals are
$p(1|\{1,2\})=\tfrac{2}{3}$,\;$p(2|\{2,3\})=\tfrac{1}{2}$ and
$q(1|\{1,2\})=\tfrac{1}{2}$,\;$q(2|\{2,3\})=\tfrac{1}{3}$.
Constructing the spliced distributions
$p_{X_1}q=\tfrac{1}{5}(2,1,2)$ (taking $p$'s ratio on $\{1,2\}$ and $q$'s
ratio on $\{2,3\}$) and $q_{X_1}p=\tfrac{1}{3}(1,1,1)$ (taking $q$'s
ratio on $\{1,2\}$ and $p$'s ratio on $\{2,3\}$), define
\[
  \pi' \;=\; \tfrac{1}{2}\,\delta_{p_{X_1}q}+\tfrac{1}{2}\,\delta_{q_{X_1}p}
       \;=\; \tfrac{1}{2}\,\delta_{(2/5,1/5,2/5)}
            +\tfrac{1}{2}\,\delta_{(1/3,1/3,1/3)}.
\]
One verifies directly that $\pi$ and $\pi'$ assign the same distribution to
every conditional: under $\{1,2\}$, both produce the multiset
$\{\tfrac{2}{3},\tfrac{1}{2}\}$ with equal weights; under $\{2,3\}$, both
produce $\{\tfrac{1}{2},\tfrac{1}{3}\}$ with equal weights.  Since
$\pi\neq\pi'$, $\pi$ is not identified.

The openness is visible from the construction: any $\pi''$ sufficiently
close to $\pi$ in the weak topology still has one atom near $p$ and one
near $q$, so the splicing argument applies to $\pi''$ as well, yielding a
distinct distribution with the same observables.  Hence the set of
non-identified distributions contains an open neighborhood of $\pi$,
consistent with the second theorem.

\begin{theorem}Suppose that $\sim_{\Sigma}$ is nonseparable and that $X\notin \Sigma$.  Then the sets  $\{\pi\in\Delta^S(\Delta_{++}(X)):\pi\text{ is identified by }\Sigma\}$ and $\{\pi\in\Delta^S(\Delta_{++}(X)):\pi\text{ is not identified by }\Sigma\}$ are both dense in $\Delta^S(\Delta_{++}(X))$.

\end{theorem}

\begin{proof}To save on notation, let $X^n = \Delta^n(\Delta_{++}(X))$ and let $Y^n = \{\pi\in\Delta^n(\Delta_{++}(X)):\pi\text{ is identified by }\Sigma\}$.  Clearly $\bigcup_n X^n$ is dense in $\Delta^S(\Delta_{++}(X))$.  Similarly, Theorem~\ref{thm:separable} shows that $Y^n$ is relatively weakly dense in $X^n$, so we know that $X^n \subseteq \overline{Y^n}$.  Therefore, $\bigcup_n X^n \subseteq \bigcup_n \overline{Y^n}$.  But clearly since for each $m$, $\overline{Y^m}\subseteq \overline{\bigcup_n Y^n}$, we get that $\bigcup_n X^n\subseteq\bigcup_n \overline{Y^n}\subseteq \overline{\bigcup_n Y^n}$, so that $\Delta^S(\Delta_{++}(X))=\overline{\bigcup_n X^n}\subseteq \overline{ \bigcup_n Y^n}$.  The result follows as every element of $\bigcup_n Y^n$ is identified.

For the density of the non-identified distributions, let $\pi$ and $\pi'$ be a pair of distributions as described in Theorem~\ref{thm:nonidentify}.  Then for any $\pi^*\in\Delta^S(\Delta_{++}(X))$, and for any $\epsilon^m\rightarrow 0$, $(1-\epsilon^m)\pi^*+\epsilon^m \pi\rightarrow \pi^*$, but $(1-\epsilon^m)\pi^* +\epsilon^m \pi'$ is a distinct member of $\Delta^S(\Delta_{++}(X))$ inducing the same distribution over marginals in $\Sigma$.  \end{proof}

\begin{theorem}Suppose that $\sim_{\Sigma}$ is separable.  Then the set $\{\pi\in\Delta^S(\Delta_{++}(X)):\pi\text{ is not identified by }\Sigma\}$ is open and dense in $\Delta^S(\Delta_{++}(X))$. \end{theorem}

\begin{proof}Let $\pi\in\Delta^S(\Delta_{++}(X))$ be in the subset for which there exists $p,q\in\mbox{supp}\pi$ for which $p(\cdot|X_1)\neq q(\cdot|X_1)$ and $p(\cdot|X_2)\neq q(\cdot|X_2)$, where $X_1$ and $X_2$ are as in the definition of separable.  Any such $\pi$ can be written as $\alpha(\frac{1}{2}\delta_p+\frac{1}{2}\delta_q)+(1-\alpha)\pi'$ for some $\pi'$; as in the proof of Theorem~\ref{thm:separable}, then $\alpha(\frac{1}{2}\delta_{p_{X_1}q}+\frac{1}{2}\delta_{q_{X_1}p})+(1-\alpha)\pi'$ is distinct from $\pi$ but induces the same measure over marginals.

It remains to establish that this set is open and dense.  Density is clear, as we may let $\pi\in\Delta^S(\Delta_{++}(X))$ be any distribution, and let $\pi^*$ any distribution for which there exists $p,q\in\mbox{supp}\pi^*$ for which $p_{X_1}\neq q_{X_1}$ and $p_{X_2}\neq q_{X_2}$.  Then for $\epsilon_m\rightarrow 0$, $(1-\epsilon_m)\pi+\epsilon_m \pi^*\rightarrow \pi$, where every element of the sequence is not identified.

Openness follows as for any such distribution $\pi$, we may consider open sets $U_p\subseteq \Delta_{++}(X)$ and $U_q\subseteq \Delta_{++}(X)$ for which $p\in U_p$, $q\in U_q$, and $\pi(U_p)=\pi(p)$, $\pi(U_q)=\pi(q)$ (that is, $U_p$ contains no member of the support of $\pi$ aside from $p$, similarly for $q$, this can be done as the support of $\pi$ is finite).  We may choose these neighborhoods small enough so that for all $p'\in U_p$ and $q'\in U_q$, $p'(\cdot|X_1)\neq q'(\cdot|X_1)$ and $p'(\cdot|X_2)\neq q'(\cdot|X_2)$.

By Urysohn (Theorem 15.6 of \citet{willard1970}), we can let $f_p$ be a continuous function equal to $1$ on $p$ and equal to $0$ on $\Delta(X)\setminus U_p$, similarly $f_q$.  Then the set of $\{\pi':\int f_p d\pi'>0\}\cap\{\pi':\int f_q d\pi'>0\}$ is a basic open neighborhood of the weak topology which contains $\pi$ and clearly any $\pi'$ in this neighborhood must have members $p'\in U_p$ and $q'\in U_q$ in its support, and thus, since $p'(\cdot|X_1)\neq q'(\cdot|X_1)$ and $p'(\cdot|X_2)\neq q'(\cdot|X_2)$, $\pi'$ is not identified.  

\end{proof}

\section{Conclusion}

Our result could also be useful for various decision theories under suitable reinterpretations. While our model studies Bayesian updating, the same mathematical foundation forms the canonical model of stochastic choice due to \citet{luce1959individual}.  See \citet{strzalecki2025} for a modern and up to date treatment. 

Under this interpretation, a distribution over Bayesian priors is equivalent to a distribution of a population's Luce weights.  The anonymous per-event posteriors correspond precisely to anonymous choice frequencies over distinct menus. Our identification question in this context is: when can an analyst recover the true population distribution over Luce weights by observing anonymous menu-by-menu choice frequencies? Following Proposition~\ref{prop:connected}, as long as the observed menu-induced graph is connected, an analyst facing a single individual can recover Luce weights up to a rescaling. Our result, under the population Luce model, shows that connectedness is not enough. Generic identification requires non-separability of the menu-induced graph (in the finite type populations); separability generically leads to non-identification. 

Under the stochastic choice interpretation, our primitive is what the discrete choice literature calls a mixed logit: a mixing
measure over Luce (logit) weights \citep{mcfaddentrain2000}.  The
decision-theoretic literature has largely asked whether some mixing measure
rationalizes the data, and how flexible the resulting model class is.
\citet{saito2018} axiomatizes the mixed logit model, and \citet{chang2024}
characterize exactly when random-coefficient models are rich enough to
approximate arbitrary random utility models; see also \citet{lu2025} on telling mixed logit apart from apart from pure characteristics models.  These
questions belong to a broader body of work characterizing the logit and Luce
family, including nested logit \citep{kovach2022nested}, focality-adjusted and perception-adjusted Luce rules \citep{kovach2022focal,
echenique2018palm}, logit as perturbed utility maximization \citep{fudenberg2015}, logit choice in dynamic preference discovery
\citep{cerreia2023},  and logit as AI's choices \citep{suleymanov2026revealed}.  We ask instead when the mixing measure is uniquely recovered, and our data
is different too.  The mixed logit literature works with the population
choice frequency on each menu: a single number for each alternative, equal
to the average choice probability once weight vectors are mixed together.
We observe more than that.  For each menu, we see the entire population
distribution of Luce choice probabilities, not just its average, so our
data refines theirs.  What we do not have is a way to link an agent's
choice probability on one menu to the same agent's choice probability on
another, since our data is anonymous across menus.  Our results show
whether the mixing measure can still be recovered without that link, and
the answer depends on the graph induced by the family of menus.

Our exercise also relates to the literature studying identification when
the data depart from the standard stochastic choice primitive.  The dataset
in \citet{azrieli2025} is a pair of distributions, one over menus and one
over alternatives: the share of choices made in each menu and the aggregate
frequency with which each alternative is chosen.  Conditional choice
frequencies are not observed at all.  Our departure runs in the opposite
direction: we observe strictly more than the conditional choice frequencies,
namely their full distribution across the population, and what is missing is
only the assignment of these conditionals to agents across menus. \citet{dardanoni2020} identify the distribution of cognitive characteristics
in a population from aggregate choice shares on a single menu.  Their
strongest result uses the joint distribution of choices across three
repetitions of the menu, and \citet{dardanoni2023} take this joint data as
the primitive, now over a collection of menus.  In their data, no individual
is identified, but the joint shares still record each agent's pattern of
choices.  In our data, even this information is missing: we observe only the
menu-by-menu marginals, and no agent's observations are linked across menus.

This failure of identification is not the familiar non-identification of the random utility model (RUM). RUM is known to be non-identified with four or more alternatives even when the analyst observes choice from all menus \citep{fishburn1998stochastic, turansick2022identification, chambers2025limits, caradonna2024identifying}, and this difficulty is unrelated to the structure of the menu family. In contrast, the population Luce model is identified once we observe all budgets by construction. The non-identification we document is a distinct phenomenon, arising specifically from anonymity across menus combined with a menu family whose graph is separable.

The welfare consequences are direct and worth stating plainly. In the population Luce model, welfare comparisons across agents, or statements about the distribution of welfare in the population, depend on knowing which weight vector a given agent holds. When the distribution of weights is not identified, these comparisons cannot be made from the data alone. Two populations that are observationally equivalent on every menu can assign entirely different welfare ranks to the same alternatives, and can disagree on which agents benefit from a policy that interferes with the choice environment.

The idea of updating a distribution over beliefs is also closely related to the prior-by-prior updating rule in the multiple priors world of \citet{gilboa1989}.  \citet{pacheco2002} and \citet{faro2019} characterize this rule axiomatically, while \citet{epstein2003} investigates a version of this rule applied to specific subclasses of events.  \citet{chambers2010} uses the same updating rule on ``objective'' sets of priors.  In our model with a finite set of ``priors,'' we can consider the support of the distribution, and the updating rule on supports is precisely the prior-by-prior updating rule.  In the multiple priors model, however, for identification purposes, one takes the set of priors to be convex---in particular, there is never room for ``observing'' priors that might be in the convex hull of the other priors.  So the question of recovering a set of priors from their Bayesian updates is fundamentally different.

Finally, we have not investigated the actual testable content of the model we consider here.  The idea here is to assume that if the distribution over $\Delta_{++}(X)$ is observable, it should be fairly easy to axiomatize the model, simply by definition.  Going further than this might prove difficult.  If we were to imagine that we only observed the distributions over $\Delta_{++}(E)$ for each $E\in \Sigma$, then, depending on whether we require our distributions to be $n$-agent distributions, or just more general Borel distributions, we may be able to answer the question by leveraging a version of Strassen's Theorem 7 \citep{strassen1965} (perhaps after taking a log transformation of the probabilities).  We leave this to future research.

%%==========================================================================
%% BIBLIOGRAPHY
%%==========================================================================

\newpage

\bibliographystyle{plainnat}
\bibliography{reference.bib}

\end{document}